\documentclass[12pt]{article}
\usepackage{epsfig,amsmath,amssymb,amsfonts,amstext,amsthm,mathrsfs}
\usepackage{latexsym,graphics,epsf,epsfig,psfrag}
\usepackage{cite}

\newtheorem{definition}{\textbf{Definition}}
\newtheorem{corollary}{\textbf{Corollary}}
\newtheorem{lemma}{\textbf{Lemma}}
\newtheorem{theorem}{\textbf{Theorem}}
\newtheorem{proposition}{\textbf{Proposition}}
\newtheorem{remark}{\textbf{Remark}}

\newcommand{\nn}{\nonumber}

\newcommand{\cU}{\mathcal{U}}

\newcommand{\cX}{\mathcal{X}}

\newcommand{\cS}{\mathcal{S}}

\newcommand{\cT}{\mathcal{T}}

\DeclareMathAlphabet{\matheuf}{U}{euf}{m}{n}

\begin{document}

\vspace*{-2cm}

\begin{center}
  \baselineskip 1.3ex {\Large \bf Bounds and Capacity Theorems for Cognitive Interference Channels with State
  \footnote{The material in this paper was presented
in part at the 49th Allerton Conference on Communication, Control, and Computing, Monticello, Illinois, USA, September
2011 and will be presented in part at the IEEE International Symposium on Information Theory, Cambridge, Massachusetts, USA, July 2012.} \footnote{The work of R. Duan and Y. Liang was supported by a National Science Foundation under Grant CCF-10-26566 and by the National Science Foundation CAREER Award under Grant
CCF-10-26565.}\\
}
 \vspace{0.15in} Ruchen Duan, Yingbin Liang
\footnote{The authors are with the Department of Electrical
Engineering and Computer Science, Syracuse University, Syracuse, NY 13244 USA (email: \{yliang06,rduan\}@syr.edu).}
\end{center}

\begin{abstract}
A class of cognitive interference channel with state is investigated, in which two transmitters (transmitters 1 and 2) communicate with two receivers (receivers 1 and 2) over an interference channel. The two transmitters jointly transmit a common message to the two receivers, and transmitter 2 also sends a separate message to receiver 2. The channel is corrupted by an independent and identically distributed (i.i.d.) state sequence. The scenario in which the state sequence is noncausally known only at transmitter 2 is first studied. For the discrete memoryless channel and its degraded version, inner and outer bounds on the capacity region are obtained. The capacity region is characterized for the degraded semideterministic channel and channels that satisfy a less noisy condition. The Gaussian channels are further studied, which are partitioned into two cases based on how the interference compares with the signal at receiver 1. For each case, inner and outer bounds on the capacity region are derived, and partial boundary of the capacity region is characterized. The full capacity region is characterized for channels that satisfy certain conditions. The second scenario in which the state sequence is noncausally known at both transmitter 2 and receiver 2 is further studied. The capacity region is obtained for both the discrete memoryless and Gaussian channels. It is also shown that this capacity is achieved by certain Gaussian channels with state noncausally known only at transmitter 2.
\end{abstract}

\section{Introduction}

Interference channels model many communication scenarios in practical wireless systems such as cellular networks, sensor networks, and cognitive radio networks. In these networks, communication between one transmitter-receiver pair may be interfered by signals from other communicating pairs which share the same spectrum resource with them. Consequently, transmission rates of these users, or in general, the throughput of a system, are affected by the strength of the interference and how the interference is treated in designing transmission schemes. Therefore, it is important to understand the fundamental communication limit (i.e., the capacity region) of interference channels. Earlier work \cite{InterferenceCarleial} by Carleial provided general bounds on the capacity region for the discrete memoryless interference channel. The achievable region was obtained by using superposition coding. Further work by Han and Kobayashi \cite{Han81} improved the achievable region via superposition and rate splitting. The capacity region of the interference channel has been characterized for various special cases, e.g., \cite{Carleial75, Sato81, ElGamal82, Costa87}. In recent a few years, some important progresses have been made on understanding the capacity region of the discrete memoryless interference channel \cite{Chong06,Kramer06,Jiang06}. In particular, new bounds on the capacity region have been derived for the Gaussian interference channel \cite{Etkin07}, which led to new capacity theorems for the Gaussian interference channel \cite{Shang09,Anna09,Mota09}. However, the capacity region of the general interference channel is still not known.

More recently, interference channels with state have caught a lot of attention. The state may be caused by many reasons such as channel uncertainty and transmitter-side signal interference. 
In particular, a few interference channel models with state noncausally known at transmitters have been studied, which are generalizations of the Gel'fand-Pinsker model \cite{GPEncoding} for the point-to-point channel with state. In \cite{Zhang11}, the interference channel with two transmitters sending two messages respectively to two receivers was studied. The channel is corrupted by an independent and identically distributed (i.i.d.) state sequence, which is noncausally known at both transmitters. A number of achievable schemes were proposed and their corresponding rate regions were compared. In \cite{Somekh08}, a model of the cognitive interference channel with state was studied, in which both transmitters (i.e., transmitters 1 and 2) jointly send one message to receiver 1, and transmitter 2 sends an additional message separately to receiver 2. The i.i.d.\ state sequence is noncausally known at transmitter 2 only. Inner and outer bounds on the capacity region were provided.

In this paper, we investigate a different class of the cognitive interference channel model with state (see Fig.~\ref{channelmodelfig}), in which both transmitters jointly send one message to both receivers 1 and 2, and transmitter 2 sends an additional message separately to receiver 2. The channel is corrupted by an i.i.d.\ state sequence. We investigate two scenarios: the first scenario assumes that the state sequence is noncausally known only at transmitter 2, and the second scenario assumes that the state sequence is known at both transmitter 2 and receiver 2. The second scenario is of interest by its own and is also useful for providing outer bounds (sometimes tight outer bounds as demonstrated in this paper) on the capacity region for the first scenario.

The difference of our model from the model studied in \cite{Somekh08} lies in that the common message known to both transmitters needs to be decoded at both receivers instead of at receiver 1 only as in \cite{Somekh08}. Although the two models appear similar to each other, their capacity regions may have different forms, and the transmission schemes achieving these regions may also be different. This is already demonstrated by the two corresponding models without state studied respectively in \cite{Wu06,Jovicic06,Maric08,Rini11,Rini12} and \cite{Liang09}. The capacity bounds in \cite{Wu06,Jovicic06} and the capacity region given in \cite{Liang09} are different and are achieved by different achievable schemes. Therefore, our study can lead to new information theoretic insights.

We note that compared to the basic Gel'fand-Pinsker model, the cognitive interference channel model we study here and in \cite{Somekh08} capture more communication features such as the transmitter-side signal cognition and receiver-side signal interference in addition to random state corruption of the channel. More specifically, transmitter 2 can be interpreted as a secondary user who knows primary user's (i.e., transmitter 1's) message $W_1$ and hence can help to transmit this message, and who also has its own message $W_2$ to transmit. The state may arise because transmitter 2 may communicate to other receivers (not necessarily receiver 1), and its signals to these receivers can be viewed as state, which is clearly known by transmitter 2. Our goal is to study the performance (i.e., the capacity region) of such a model and correspondingly design communication schemes to exploit the noncausal state information in the context of signal cognition and interference.

In the following, we summarize the main results of this paper. We note that due to the channel properties of cognition, interference, random channel state, and asymmetry of the state knowledge, it is natural that an achievable scheme employs coding techniques of superposition, rate splitting, and Gel'fand-Pinsker coding. The novelty of this paper lies in finding optimality of such achievable schemes (i.e., achievement of the capacity region) by properly integrating these coding techniques for various channel parameters. The new gradients that we develop in the converse arguments are also mentioned below.

For the discrete memoryless cognitive interference channel with noncausal state information known at transmitter 2, we derive inner and outer bounds on the capacity region. In particular, due to asymmetry of the state knowledge (i.e., transmitter 1 does not know the channel state but transmitter 2 does), transmitter 2 not only helps transmitter 1 in the conventional way of superposition, but also helps to correlate the input with the state sequence via Gel'fand-Pinsker scheme. Thus, we employ the Gel'fand-Pinsker scheme for these two cooperative transmitters in the way that transmitter 1 generates signals with only the message index, superposing on which transmitter 2 generates auxiliary variables with the bin index. We show by special cases that such an auxiliary variable is necessary to achieve the capacity.

We then study the degraded channel of the model, and obtain bounds on the capacity region. It is not surprising that the capacity region for the degraded channel is not obtained because it is difficult to obtain the capacity region even for the degraded broadcast channel with state \cite{BroadcastSteinberg}. However, we establish the capacity region for degraded channels, which further satisfy the semideterminsitic condition. This example channel also demonstrates that both superposition and Gel'fand-Pinsker coding for state treatment in transmitter 2's cooperation are necessary for achieving the capacity. Besides the semideterministic degraded channel, we also identify a less noisy condition under which we obtain the capacity region.

We further study the Gaussian channel of the model. Although for the Gaussian channel, it is natural to obtain an achievable region by applying the general jointly Gaussian input distribution to the inner bound derived for the discrete memoryless channel, the resulting region would have a too complex form. It would then be very difficult to develop a converse proof for capacity characterization. Our approach is to partition the Gaussian channel into two cases depending on how the interference compares with the signal at receiver 1. For each case, we develop simpler inner bounds that exploit the conditions that the channel satisfies. For such inner bounds, we are able to derive outer bounds that match the inner bounds for partial boundary of the capacity region.

More specifically, for the first Gaussian case when the channel gain of interference is stronger than the channel gain of signal at receiver 1, it is reasonable to let receiver 1 decode full information intended for receiver 2. We derive inner bound based on such a scheme. We also provide an outer bound and further identify rate points that inner and outer bounds match at the boundary. These points hence characterize partial boundary of the capacity region. We also identify a condition, under which the outer bound fully characterizes the capacity region.

For the second Gaussian case when the channel gain of interference is weaker than the channel gain of signal at receiver 1, rate splitting is also not necessary but with receiver 1 decoding no information intended for receiver 2. Hence, without using rate splitting, we obtain two inner bounds with the Gel'fand-Pinsker scheme canceling the state respectively at receivers 1 and 2. Similarly to the first Gaussian case, for each inner bound, we provide an outer bound and identify rate points that the inner and outer bounds match at the boundary. We further show that respectively under two channel conditions, each outer bound characterizes the full capacity region. In particular, one of these conditions leads to the case that the Gaussian channel with state known only at transmitter 2 achieves the capacity region of the Gaussian channel with state known at both transmitter 2 and receiver 2. This is similar to the case that dirty paper coding achieves the capacity of the Gaussian channel when the state is also known at the receiver \cite{Costa83}. Here, the channel does not achieve the capacity with both receivers knowing the channel state due to asymmetry of the state knowledge at the transmitter side.

We finally study the cognitive interference channel with the state noncausally known at both transmitter 2 and receiver 2. For this scenario, we characterize the full capacity region for both the discrete memoryless and Gaussian channels. For the discrete memoryless channel, we first derive inner and outer bounds on the capacity region, which are characterized by different forms. Standard techniques do not provide an easy argument of the equivalence of the two bounds. We apply the technique recently developed by Lapidoth and Wang in \cite{Lapidoth11} for proving equivalence of two rate regions characterized by different sets of auxiliary random variables, and show that our inner and outer bounds match. For Gaussian channels, we also partition them into two cases, and characterize the full capacity region for each case. In particular, the converse argument involves specially designed state knowledge for receivers such that the resulting outer bounds are tight. Such construction is inspired by the fact that dirty paper coding achieves the capacity of the Gaussian channel when the state is also known at the receiver \cite{Costa83}.


%

The rest of the paper is organized as follows. In Section \ref{sec:channelmodel}, we describe the channel model and explain the notation used in this paper. In Sections \ref{sec:mainresult} and \ref{sec:generalGaussian}, we present our results for the discrete memoryless channel and Gaussian channel, respectively, for the scenario with the state known at transmitter 2 only. In Sections \ref{sec:recstate}, we present the results for the scenario with the state also known at receiver 2 for both the discrete memoryless and Gaussian channels. Finally, in Section \ref{sec:conclusion}, we conclude with a few remarks.

\section{Channel Model}\label{sec:channelmodel}

\begin{figure}[thb]
\centering
\includegraphics[width=3.7in]{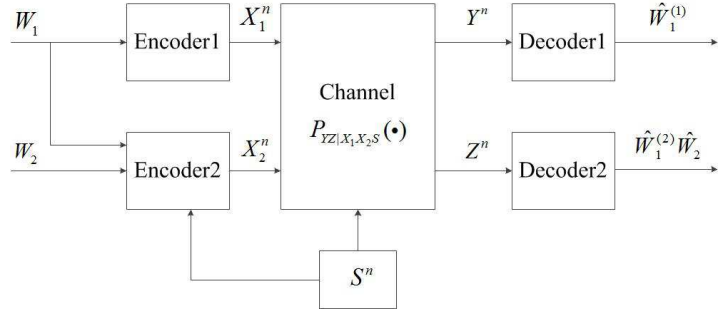}
\caption{A model of the cognitive interference channel with state}\label{channelmodelfig}
\label{fig:model}
\end{figure}

We consider a class of cognitive interference channels with state (see Fig.~\ref{channelmodelfig}), in which two transmitters (say transmitters 1 and 2) jointly send a message $W_{1}$ to two receivers (say receivers 1 and 2), and transmitter 2 sends a message $W_{2}$ to receiver 2. The channel is also corrupted by an i.i.d.\ state sequence $S^n$. We investigate two scenarios: the first scenario assumes that the state sequence is noncausally known only at transmitter 2 but not known at any other terminal, and the second scenario assumes that the state sequence is known at both transmitter 2 and receiver 2.
In this paper, we use $s^n$ to denote the vector $(s_1,\dotsi,s_n)$, and use $s_{i}^n$ to denote the vector $(s_i,\dotsi,s_n)$. We formally define the channel model as follows.

\begin{definition}
A discrete memoryless cognitive interference channel with state consists of two finite channel input alphabets $\mathcal{X}_{1}$ and $\mathcal{X}_{2}$, a finite state alphabet $\mathcal{S}$, two finite channel output alphabets $\mathcal{Y}$ and $\mathcal{Z}$, and a transition probability distribution $P_{Y Z \vert X_{1} X_{2} S}$ (see Fig.~\ref{channelmodelfig}), where $X_{1} \in \mathcal{X}_{1}$ and $X_{2} \in \mathcal{X}_{2}$ are the channel inputs from transmitters 1 and 2, respectively, $S \in \mathcal{S}$ is the state variable, and $Y \in \mathcal{Y}$ and $Z \in \mathcal{Z}$ are the channel outputs at receivers 1 and 2, respectively.
\end{definition}
\begin{definition}
A $(2^{nR_1},2^{nR_2},n)$ code for the cognitive interference channel with state noncausally known only at transmitter 2 consists of the following:
\begin{list}{$\bullet$}{\topsep=0.ex \leftmargin=0.25in
\rightmargin=0.in \itemsep =0.in}

\item two message sets: $\mathcal{W}_k = {1,2,\dotsi,2^{nR_k}}$ for $k=1,2$;

\item two messages: $W_1$ and $W_2$ are independent random variables and are uniformly distributed over $\mathcal{W}_1$ and $\mathcal{W}_2$, respectively;

\item two encoders: an encoder $f_1: \mathcal{W}_1 \to \mathcal{X}_1^n$, which maps a message $w_1 \in \mathcal{W}_1$ to a codeword $x_1^n \in \mathcal{X}_1^n$; and an encoder $f_2: \mathcal{W}_1 \times
    \mathcal{W}_2 \times S^n \to \mathcal{X}_2^n$, which maps a message pair $(w_1,w_2) \in \mathcal{W}_1 \times \mathcal{W}_2$ and a state sequence $s^n \in S^n$ to a codeword $x_2^n \in
    \mathcal{X}_2^n$;

\item two decoders: $g_1: \mathcal{Y}^n \to \mathcal{W}_1 $, which maps a received sequence $y^n$ into a message $\hat{w}_1^{(1)} \in \mathcal{W}_1$; and $g_2:  \mathcal{Z}^n \to \mathcal{W}_1 \times
    \mathcal{W}_2$, which maps a received sequence $z^n$ into a message pair $\left(\hat{w}_1^{(2)},\hat{w}_2\right) \in \mathcal{W}_1 \times \mathcal{W}_2$.

\end{list}
\end{definition}
We note that the above definition is also applicable to the scenario with the state sequence known at both transmitter 2 and receiver 2, if the second decoder is changed to $g_2:  (\mathcal{Z}^n, S^n) \to \mathcal{W}_1 \times \mathcal{W}_2$.

For a given code, we define the probability of error as
\begin{equation} \label{PE}
\begin{split}
P_e^{(n)} =  \frac{1}{2^{n(R_1+R_2)}} \sum_{w_1=1}^{2^{nR_1}} \sum_{w_2=1}^{2^{nR_2}} Pr\left\{\left(\hat{w}_1^{(1)},\hat{w}_1^{(2)},\hat{w}_2 \right) \neq (w_1,w_1,w_2)\right\} \textrm{.}
\end{split}
\end{equation}
A rate pair $(R_1,R_2)$ is said to be \textit{achievable} if there exists a sequence of $(2^{nR_1},2^{nR_2},n)$ codes such that
\begin{equation}
\lim_{n \to \infty} P_e^{(n)} = 0\textrm{.}
\end{equation}
\begin{definition}
The capacity region is defined to be the closure of the set of all achievable rate pairs $(R_1,R_2)$.
\end{definition}

In the following, we define a number of channel conditions for classifying the channels in our study:
\begin{flalign}
& \bullet \: P_{YZ|X_1X_2S}=P_{Z|X_1X_2S}P_{Y|Z} \label{eq:cond5} \\
& \bullet \; P_{YZ|X_1X_2S}=P_{Z|X_1X_2S}P_{Y|ZX_1S} \label{eq:cond1} \\
& \bullet \; P_{YZ|X_1X_2S}=P_{Y|X_1X_2S}P_{Z|YX_1S} \label{eq:cond2} \\
& \bullet \; I(X_1;Y) \leq I(X_1;Z)\textrm{ } and \textrm{ } I(U;Y|X_1) \leq I(U;Z|X_1) \nonumber \\
& \quad\quad \text{ for all } P_{UX_1X_2S} \; \text{ s.t. } \; P_{X_1SUX_2YZ} = P_{X_1}P_{S}P_{UX_2|SX_1}P_{YZ|SX_1X_2} \label{eq:cond3} \\
& \bullet \; I(X_1U;Y) \ge I(X_1U;Z) \nn \\
& \quad\quad \text{ for all } P_{UX_1X_2S} \; \text{ s.t. } \; P_{X_1SUX_2YZ} = P_{X_1}P_{S}P_{UX_2|SX_1}P_{YZ|SX_1X_2} \label{eq:cond4}
\end{flalign}
The intuitive meaning of the above conditions are explained as follows. If a channel satisfies \eqref{eq:cond5}, receiver 2 is stronger than receiver 1 in decoding $W_1$ and $W_2$. If a channel satisfies \eqref{eq:cond1}, receiver 2 is stronger in decoding $W_2$, which is weaker than condition \eqref{eq:cond5}. In contrast to \eqref{eq:cond1}, if a channel satisfies \eqref{eq:cond2}, receiver 1 is stronger in decoding $W_2$, although this message is not intended for receiver 1. If a channel satisfies
\eqref{eq:cond3}, receiver 2 is less noisy than receiver 1 in the sense similar to the less noisy condition defined in \cite{Korner75}. Alternatively, if a channel satisfies \eqref{eq:cond4}, receiver 1 is less noisy than receiver 2.

\section{Discrete Memoryless Channels}\label{sec:mainresult}

In this section, we investigate the discrete memoryless cognitive interference channel with state noncausally known at only transmitter 2. We first provide inner and outer bounds on the capacity region, and then we identify a few special cases, for which we establish the capacity region.

We first provide an achievable region in the following lemma, which is useful in establishing our main inner bound.
\begin{lemma} \label{NaiAchv}
An achievable region for the cognitive interference channel with the state sequence noncausally known at transmitter 2 consists of rate pairs $(R_1, R_2)$ satisfying:
\begin{flalign}
      R_2 & = R_{21} + R_{22}, \quad R_{21}\geqslant 0, \quad R_{22}\geqslant 0\nn\\
      R_1 + R_{21} & \leqslant I(TUX_1;Y) - I(TU;S|X_1)\nn\\
      R_{22} & \leqslant I(V;Z|UTX_1) - I(V;S|UTX_1)\nn\\
      R_{21} + R_{22} & \leqslant I(UV;Z|X_1T) - I(UV;S|X_1T)\nn\\
      R_{21} + R_{22} & \leqslant I(T UV;Z|X_1)- I(TUV;S|X_1)\nn\\
      R_1 + R_{21} + + R_{22} & \leqslant I(TUVX_1;Z) - I(TUV;S|X_1)
\end{flalign}
for some distribution $P_{X_1STUVX_2YZ} = P_{X_1}P_{S}P_{TUVX_2|SX_1}P_{YZ|SX_1X_2}$, where $T$, $U$ and $V$ are auxiliary random variables.
\end{lemma}
\begin{proof}
The achievable scheme includes superposition coding, rate-splitting, and Gel'fand-Pinsker binning scheme. We outline the achievable scheme as follows. Transmitter 1 first encodes $W_1$. Transmitter 2 cooperatively transmits $W_1$. Due to asymmetry of the state knowledge (i.e., transmitter 1 does not know the channel state but transmitter 2 does), transmitter 2 not only helps transmitter 1 using superposition, but also helps in correlating the input with the state sequence via the Gel'fand-Pinsker scheme. Thus, for transmitting $W_1$, transmitter 1 generates $X_1$ with only the message index, superposing on which transmitter 2 generates an auxiliary random variable $T$ with the bin index.

Then transmitter 2 employs rate splitting for transmitting $W_2$. Namely, message $W_2$ is split into two components, $W_{21}$ and $W_{22}$, with rates $R_{21}$ and $R_{22}$, respectively. The message $W_{21}$ is intended for both receivers to decode, and $W_{22}$ is intended only for receiver 2 to decode. Transmitter 2 encodes $W_{21}$ and $W_{22}$ by superposing on $W_1$. Furthermore, transmitter 2 uses Gel'fand-Pinsker scheme for correlating the inputs (that encode $W_{21}$ and $W_{22}$, and are respectively represented by $U$ and $V$ in the above region) with the state sequence. Receiver 1 decodes both $W_1$ and $W_{21}$, and receiver 2 decodes $W_1$, $W_{21}$ and $W_{22}$. Since receiver 1 can decode $W_{21}$, it can eliminate the interference caused by this message when it decodes $W_1$.

The detailed proof is relegated to Appendix \ref{ProofNaiAch}.
\end{proof}
In the above achievable schemes, it is seemingly true that the role of $T$ can be performed by $U$, and may not be necessary, because they both represent messages intended for both receivers. However, we show by special cases that $T$ is necessary for achieving the capacity region but $U$ may be removed (i.e., rate splitting is unnecessary).

Based on Lemma \ref{NaiAchv}, our main inner bound on the capacity region is given in the following theorem.
\begin{theorem}\label{tr:OrigAch}(Inner Bound)
For the cognitive interference channel with the state sequence noncausally known at transmitter 2, an achievable region consists of rate pairs $(R_1, R_2)$ satisfying:
\begin{flalign}
R_1 \leqslant & I(X_1TU;Y) - I(TU;S|X_1)\nn\\
R_2 \leqslant & I(UV;Z|X_1T) - I(UV;S|X_1T)\nn\\
R_2 \leqslant & I(TUV;Z|X_1) - I(TUV;S|X_1)\nn\\
R_1 + R_2 \leqslant & I(X_1TUV;Z) - I(TUV;S|X_1)\nn\\
R_1 + R_2 \leqslant & I(X_1TU;Y) + I(V;Z|X_1TU) - I(TUV;S|X_1)\label{OrigAch}
\end{flalign}
for some distribution $P_{X_1STUVX_2YZ} = P_{X_1}P_{S}P_{TUVX_2|SX_1}P_{YZ|SX_1X_2}$ that satisfies
\begin{equation}\label{eq:cond}
I(V;Z|UTX_1)-I(V;S|UTX_1)\ge 0.
\end{equation}
\end{theorem}
\begin{proof}
By applying Fourier-Motzkin elimination\cite{FMElimination}, we eliminate $R_{21}$ and $R_{22}$ from the bounds in Lemma \ref{NaiAchv} and obtain the bounds in Theorem \ref{tr:OrigAch}.
\end{proof}
We note that the condition \eqref{eq:cond} follows from Fourier-Motzkin elimination to guarantee validness of the region in Lemma \ref{NaiAchv}.
\begin{remark}
  The achievable region in Theorem \ref{tr:OrigAch} reduces to the capacity region of the multiple-access channel with state known noncausally at one transmitter in \cite{MACSomekhBaruch08} by setting $Y=Z$, $T=\phi$ and $V=U$.
\end{remark}
\begin{remark}
  The achievable region in Theorem \ref{tr:OrigAch} reduces to the capacity region of the cognitive interference channel without state in \cite{Liang09} by setting $S=\phi$, $T=\phi$ and $V=X_2$.
\end{remark}

Following Theorem \ref{tr:OrigAch}, we derive the following inner bound by setting $U=\phi$, which is achieved by a scheme without rate splitting. This inner bound is useful for studying Gaussian channels in Section \ref{sec:gauss1}.
\begin{corollary}\label{cor:innerdm2}(Inner Bound)
For the cognitive interference channel with state noncausally known at transmitter 2, an achievable region consists of rate pairs $(R_1, R_2)$ satisfying:
\begin{flalign}
R_1 \leqslant & I(X_1T;Y) - I(T;S|X_1)\nn\\
R_2 \leqslant & I(V;Z|X_1T) - I(V;S|X_1T)\nn\\
R_2 \leqslant & I(TV;Z|X_1) - I(TV;S|X_1)\nn\\
R_1 + R_2 \leqslant & I(X_1TV;Z) - I(TV;S|X_1)\label{eq:DMCinner21}
\end{flalign}
for some distribution $P_{X_1STVX_2YZ} =P_{X_1}P_{S}P_{TVX_2|X_1S}P_{YZ|SX_1X_2}$ that satisfies
\begin{equation}
I(V;Z|TX_1)-I(V;S|TX_1)\ge 0.
\end{equation}
\end{corollary}


We next provide an outer bound on the capacity region for the cognitive interference channel with state.
\begin{theorem}\label{thr:OrigAchConv}(Outer Bound)
An outer bound for the interference channel with state noncausally known at transmitter 2 consists of the rate pairs $(R_1, R_2)$ satisfying:
\begin{flalign}
R_1 \leqslant & I(X_1TU;Y) - I(T U;S|X_1)\nn\\
R_2 \leqslant & I(TV;Z|X_1) - I(TV;S|X_1)\nn\\
R_1 + R_2 \leqslant & I(X_1TV;Z) - I(TV;S|X_1)
\end{flalign}
for some distribution $P_{X_1STUVX_2YZ} = P_{X_1}P_{S}P_{TUVX_2|X_1S}P_{YZ|SX_1X_2}$, which satisfies the Markov chain conditions $T \leftrightarrow UV \leftrightarrow X_{1}X_{2}S \leftrightarrow YZ$.
\end{theorem}
\begin{proof}
The proof employs the techniques in \cite{GPEncoding} for the Gel'fand-Pinsker model, and exploits independence properties among variables in our model. In particular, the auxiliary random variables are carefully constructed. The detailed proof is relegated to Appendix \ref{ProofOrigConv}.
\end{proof}

We now provide inner and outer bounds for the degraded channel, which are useful for further identifying the cases for which we obtain the capacity region.
\begin{theorem}\label{DegrConv}(Inner and Outer Bounds)
If the cognitive interference channel with the state sequence noncausally known at transmitter 2 satisfies the degradedness condition \eqref{eq:cond5} (i.e., receiver 1 is degraded with regard to receiver 2), then an achievable region consists of the rate pairs $(R_1, R_2)$ satisfying:
\begin{flalign}\label{eq:DegrAchi}
R_1 \leqslant & I(X_1T;Y) - I(T;S|X_1)\nn\\
R_2 \leqslant & I(V;Z|X_1T) - I(V;S|X_1T)\nn\\
R_2 \leqslant & I(TV;Z|X_1) - I(TV;S|X_1)
\end{flalign}
for some distribution $P_{X_1STVX_2YZ} = P_{X_1}P_{S}P_{TVX_2|X_1S}P_{YZ|SX_1X_2}$ that satisfies
\begin{equation}
I(V;Z|TX_1)-I(V;S|TX_1)\ge 0.
\end{equation}

An outer bound on the capacity region for such a channel consists of the rate pairs $(R_1, R_2)$ satisfying:
\begin{flalign}\label{eq:DegrConv}
R_1 \leqslant & I(X_1 T;Y) - I(T;S|X_1)\nn\\
R_2 \leqslant & I(TV;Z|X_1) - I(TV;S|X_1)
\end{flalign}
for some distribution $P_{X_1STVX_2YZ} = P_{X_1}P_{S}P_{TVX_2|X_1S}P_{YZ|SX_1X_2}$, which satisfies the Markov chain conditions $T \leftrightarrow V \leftrightarrow X_{1}X_{2}S \leftrightarrow YZ$.
\end{theorem}
\begin{proof}
The achievability follows from the achievable region given in Corollary \ref{cor:innerdm2} by removing the fourth bound on $R_1+R_2$ due to the degradedness condition. The proof of the outer bound is detailed in Appendix \ref{pr:DegrConv}.
\end{proof}
\begin{remark}
By setting $X_1 = \phi$, the third bound in \eqref{eq:DegrAchi} is redundant, and the achievable region in Theorem \ref{DegrConv} coincides with the achievable region for the degraded broadcast channel with state noncausally known at the transmitter in \cite{BroadcastSteinberg}. This is reasonable because although the model in \cite{BroadcastSteinberg} does not require receiver 2 to decode $W_1$ as in our model, receiver 2 is able to do so due to the degradedness condition.
\end{remark}

The inner and outer bounds given in Theorems \ref{tr:OrigAch} and \ref{thr:OrigAchConv} do not match in general. We next identify two classes of channels, for which we obtain the capacity region. We first provide the capacity region for the degraded semideterministic channel in the following theorem.
\begin{theorem}\label{DetmDegrConv}(Capacity)
If the interference channel with the state sequence noncausally known at transmitter 2 satisfies the degradedness condition \eqref{eq:cond5} and the semideterministic condition such that $P_{Z|X_1X_2S}$ takes on values of either ``$0$" or ``$1$", then the capacity region of the channel consists of rate pairs $(R_1, R_2)$ satisfying:
\begin{flalign}
R_1 & \leqslant I(X_1T;Y) - I(T;S|X_1)\nn\\
R_2 & \leqslant H(Z|X_1TS) \nn \\
R_2 & \leqslant H(Z|X_1) - I(TZ;S|X_1)\label{eq:semi3}
\end{flalign}
for some distribution $P_{X_1STX_2YZ}=P_{X_1}P_SP_{TX_2|SX_1}P_{Z|X_1X_2S}P_{Y|Z}$, where $T$ is an auxiliary random variable and its cardinality is bounded by $|\cT|\leqslant |\cX_1||\cX_2||\cS|+1$.
\end{theorem}
\begin{proof}
The achievability follows from the achievable region in \eqref{eq:DegrAchi} by setting $V=Z$. The proof of the converse is detailed in Appendix \ref{ProofDetmDegrConv}.
\end{proof}

Following Theorem \ref{DetmDegrConv}, we also obtain the capacity region for the semideterministic degraded broadcast channel with the noncausal state information known at the transmitter by setting $X_1=\Phi$ in Theorem \ref{DetmDegrConv}, which consists of rate pairs $(R_1, R_2)$ satisfying:
  \begin{flalign}
      R_1 \leqslant & I(T;Y) - I(T;S) \label{eq:semi2}\\
      R_2 \leqslant & H(Z|TS)
  \end{flalign}
for some distribution that $P_{STXYZ}=P_SP_{TX|S}P_{Z|XS}P_{Y|Z}$, where $X$ is the channel input, and $Y$ and $Z$ are the channel outputs. We note that the third bound in \eqref{eq:semi3} becomes redundant when setting $X_1=\Phi$, because
\begin{flalign}
H(Z) - I(TZ;S)&=H(Z|TS)+(I(T;Z)-I(T;S)) \nn \\
&\ge H(Z|TS)+(I(T;Y)-I(T;S)) \nn \\
& \ge H(Z|TS)
\end{flalign}
where $I(T;Y)-I(T;S) \ge 0$ is necessary to guarantee $R_1 \ge 0 $ in \eqref{eq:semi2}.

We next obtain the following capacity region when receiver 1 is less noisy than receiver 2, i.e, the channel satisfies the condition \eqref{eq:cond4}.
\begin{theorem}(Capacity)
For the cognitive interference channel with state noncausally known at transmitter 2, if it satisfies the condition \eqref{eq:cond4}, the capacity region consists of rate pairs $(R_1, R_2)$ satisfying:
\begin{equation}
\begin{split}
R_2 \leqslant & I(U;Z|X_1) - I(U;S|X_1)\\
R_1 + R_2 \leqslant & I(X_1U;Z) - I(U;S|X_1)
\end{split}
\end{equation}
for some distribution $P_{X_1SUX_2YZ} =P_{X_1}P_{S}P_{UX_2|X_1S}P_{YZ|SX_1X_2}$, where $U$ is an auxiliary random variable and its cardinality is bounded by $|\cU|\leqslant |\cX_1||\cX_2||\cS|$.
\end{theorem}
\begin{proof}
Achievability follows from Theorem 1 by setting $T=\phi$, $V=U$ and using \eqref{eq:cond4} to remove the redundant bounds. The converse follows from the capacity region of the multiple access channel (with its receiver being receiver 2 in our model) with state available at one transmitter given in \cite{MACSomekhBaruch08}, which clearly is an outer bound for our model.
\end{proof}

\section{Gaussian Channels}\label{sec:generalGaussian}

In this section, we consider the Gaussian cognitive interference channel with state noncausally known at only transmitter 2. The channel outputs at receivers 1 and 2 at time instant $i$ are given by
\begin{flalign}
  &Y_i=X_{1i}+a X_{2i}+S_i+N_{1i} \nn\\
&Z_i=b X_{1i}+X_{2i}+c S_i+N_{2i} \label{eq:gaussmodel}
\end{flalign}
where the noise variables $N_{1i}\sim \mathcal{N}(0,1)$ and $N_{2i} \sim \mathcal{N}(0,1)$, and the state variable $S_i \sim \mathcal{N}(0,Q)$. Both the noise variables and the state
variable are i.i.d. over channel uses. As we assume for the discrete memoryless channel, the state sequence $\{S_i\}_{i=1}^n$ is noncausally known at transmitter 2 only. The channel inputs are subject to the average power constraints
\begin{equation}
  \frac{1}{n} \sum_{i=1}^n X_{1i}^2\leqslant P_1 \quad \textrm{and}\quad \frac{1}{n} \sum_{i=1}^n X_{2i}^2\leqslant P_2 \textrm{.}
\end{equation}

We partition the Gaussian cognitive interference channel with state into two classes corresponding to $|a| \leqslant 1$ and $|a| > 1$, and study these two classes separately in this and next subsections. In each subsection, we first provide inner and outer bounds on the capacity region, and then characterize partial boundaries of the capacity region based on these bounds. We also obtain the full capacity region for channels that satisfy certain conditions.

We note that our results for Gaussian channels exploit the fact that for both $|a| > 1$ and $|a| \leqslant 1$, the Gaussian channel is stochastically degraded given $X_1$ and $S$, i.e., its marginal distributions at the two receivers are the same as a physically degraded Gaussian channel that satisfies the condition \eqref{eq:cond2} and \eqref{eq:cond1}, respectively. Because the capacities of the two Gaussian channels are the same, our results below are applicable to both stochastically degraded and physically degraded channels with the proofs exploiting the physical degradedness conditions \eqref{eq:cond2} and \eqref{eq:cond1}.

\subsection{Gaussian Channel: $|a| > 1$}\label{sec:gauss2}

\subsubsection{Inner and Outer Bounds}

If $|a| > 1$, the Gaussian channel satisfies the condition (\ref{eq:cond2}). We first provide an inner bound for this class of channels.
\begin{proposition}(Inner Bound)\label{th:inner21}
For the Gaussian cognitive interference channel with state noncausally known at transmitter 2, if $|a|>1$, an inner bound consists of rate pairs $(R_1,R_2)$ satisfying:
\begin{small}
\begin{flalign}
R_2 \leqslant & \frac{1}{2} \log(1+P_2') \nn \\
R_1+R_2 \leqslant & \frac{1}{2} \log\left(1+\frac{b^2P_1+2b\rho_{21}\sqrt{P_1P_2}+\rho_{21}^2P_2}{(1-\rho_{21}^2)P_2+2c\rho_{2s}\sqrt{P_2Q}+c^2Q+1}\right)+\frac{1}{2} \log(1+P_2') \nn \\
R_1+R_2 \leqslant & \frac{1}{2} \log\left(1+\frac{P_1+2a\rho_{21}\sqrt{P_1P_2}+a^2\rho_{21}^2P_2}{a^2(1-\rho_{21}^2)P_2+2a\rho_{2s}\sqrt{P_2Q}+Q+1}\right) \nn \\
& +\frac{1}{2}\log \left(1+\frac{a^2P_2'^2+2a\rho_{2s1}\rho_{2s2}P_2'-a^2\rho_{2s1}^2P_2'-\rho_{2s1}^2}{a^2\rho_{2s1}^2P_2'+\rho_{2s2}^2P_2'+P_2'+\rho_{2s1}^2-2a\rho_{2s1}\rho_{2s2}P_2'} \right)\label{eq:inner211}
\end{flalign}
\end{small}
where $P_2' = (1-\rho_{21}^2-\rho_{2s}^2)P_2$ and $\rho_{21}^2+\rho_{2s}^2 \leqslant 1$, $\rho_{2s1} = \alpha(c\sqrt{Q}+\rho_{2s}\sqrt{P_2})$, $\rho_{2s2}=(\sqrt{Q}+a\rho_{2s}\sqrt{P_2})$, $\alpha=\frac{P_2'}{P_2'+1}$.
\end{proposition}
\begin{proof}
By setting $T=\phi$ and $U=V$ in the inner bound given in Theorem \ref{tr:OrigAch}, we obtain an inner
bound that includes the following bounds:
\begin{flalign}
R_2 \leqslant & I(U;Z|X_1) - I(U;S|X_1)\nn\\
R_1 + R_2 \leqslant & I(X_1U;Z) - I(U;S|X_1)\nn\\
R_1 + R_2 \leqslant & I(X_1U;Y)-I(U;S|X_1)\; .\label{eq:DMCinner2}
\end{flalign}
Based on the above bounds, we choose the jointly Gaussian input distribution and employ dirty paper coding for $U$ to deal with the state in $Z$. More specifically, we set the random variables as follows and obtain the desired inner bound:
  \begin{flalign}
      &X_1 \sim \mathcal{N}(0,P_1), \quad X_2' \sim \mathcal{N}(0,P_2'), \quad P_2'  = (1-\rho_{21}^2-\rho_{2S}^2)P_2 \nn\\
      &X_2 = \rho_{21} \sqrt{\frac{P_2}{P_1}}X_1+ X_2'+ \rho_{2s} \sqrt{\frac{P_2}{Q}}S \nn\\
      &U = X_2' + \alpha \left(c + \rho_{2s}\sqrt{\frac{P_2}{Q}}\right)S \label{eq:innerbound2}
  \end{flalign}
where $X_1$, $X_2'$ and $S$ are independent random variables, and $\alpha = \frac{P_2'}{P_2' + 1}$.
\end{proof}

We next provide an outer bound on the capacity region.
\begin{proposition}(Outer Bound)\label{th:outer21}
For the Gaussian cognitive interference channel with state noncausally known at transmitter 2, if $|a|>1$, an outer bound consists of rate pairs $(R_1,R_2)$ satisfying:
\begin{small}
\begin{flalign}
R_2 \leqslant & \frac{1}{2} \log(1+P_2') \nonumber\\
R_1 + R_2 \leqslant & \frac{1}{2} \log\left(1+\frac{b^2P_1+2b\rho_{21}\sqrt{P_1P_2}+\rho_{21}^2P_2}{(1-\rho_{21}^2)P_2+2c\rho_{2s}\sqrt{P_2Q}+c^2Q+1}\right) + \frac{1}{2} \log(1+(1-\rho_{21}^2-\rho_{2s}^2)P_2) \label{eq:outer211}
  \end{flalign}
  \end{small}
where $P_2' \leq (1-\rho_{21}^2-\rho_{2s}^2)P_2$ and $\rho_{21}^2+\rho_{2s}^2 \leqslant 1$.
\end{proposition}
\begin{proof}
It is clear that the outer bound in Proposition \ref{th:outer21} is equivalent to the region that consists of rate pairs $(R_1,R_2)$ satisfying:
\begin{flalign}
R_2 \leqslant & \frac{1}{2} \log(1+(1-\rho_{21}^2-\rho_{2s}^2)P_2) \nonumber\\
R_1 + R_2 \leqslant & \frac{1}{2} \log\left(1+\frac{b^2P_1+2b\rho_{21}\sqrt{P_1P_2}+\rho_{21}^2P_2}{(1-\rho_{21}^2)P_2+2c\rho_{2s}\sqrt{P_2Q}+c^2Q+1}\right) + \frac{1}{2} \log(1+(1-\rho_{21}^2-\rho_{2s}^2)P_2) \label{eq:macouter}
\end{flalign}
where $\rho_{21}^2+\rho_{2s}^2 \leqslant 1$. This region is the capacity region of the multiple access channel with state (with its receiver being receiver 2 in our model) given in \cite{MACSomekhBaruch08}, and hence serves as an outer bound for our model.
\end{proof}
We note that although the region \eqref{eq:outer211} is equivalent to the region \eqref{eq:macouter}, the form of \eqref{eq:outer211} is more convenient for characterizing the boundary points of the capacity region in the following subsection.

\subsubsection{Capacity Theorem}

Although the inner bound \eqref{eq:inner211} and the outer bound \eqref{eq:outer211} do not match in general, we show that these bounds characterize some boundary points of the capacity region. We also show that the outer bound characterize the full capacity region if the channel satisfies certain conditions.

In order to characterize the boundary points of the capacity region, we first change the inner bound in \eqref{eq:inner211} to a more convenient form, which consists of rate pairs $(R_1,R_2)$ satisfying:
\begin{small}
\begin{flalign}
R_2 \leqslant & \frac{1}{2} \log(1+P_2') \label{eq:inner211-1} \\
R_1+R_2 \leqslant & \frac{1}{2} \log\left(1+\frac{b^2P_1+2b\rho_{21}\sqrt{P_1P_2}+\rho_{21}^2P_2}{(1-\rho_{21}^2)P_2+2c\rho_{2s}\sqrt{P_2Q}+c^2Q+1}\right)+\frac{1}{2} \log(1+(1-\rho_{21}^2-\rho_{2s}^2)P_2) \label{eq:inner211-2} \\
R_1+R_2 \leqslant & \frac{1}{2} \log\left(1+\frac{P_1+2a\rho_{21}\sqrt{P_1P_2}+a^2\rho_{21}^2P_2}{a^2(1-\rho_{21}^2)P_2+2a\rho_{2s}\sqrt{P_2Q}+Q+1}\right) \nn \\
& \hspace{-1.5cm}+\frac{1}{2}\log \left(1+\frac{\left(a^2(1-\rho_{21}^2-\rho_{2s}^2)P_2+2a\rho_{2s1}\rho_{2s2}-a^2\rho_{2s1}^2\right)(1-\rho_{21}^2-\rho_{2s}^2)P_2-\rho_{2s1}^2}
{\left(a^2\rho_{2s1}^2+\rho_{2s2}^2+1-2a\rho_{2s1}\rho_{2s2}\right)(1-\rho_{21}^2-\rho_{2s}^2)P_2+\rho_{2s1}^2} \right)\label{eq:inner211-3}
\end{flalign}
\end{small}
where $P_2' \leq (1-\rho_{21}^2-\rho_{2s}^2)P_2$ and $\rho_{21}^2+\rho_{2s}^2 \leqslant 1$, $\rho_{2s1} = \alpha(c\sqrt{Q}+\rho_{2s}\sqrt{P_2})$, $\rho_{2s2}=(\sqrt{Q}+a\rho_{2s}\sqrt{P_2})$, $\alpha=\frac{(1-\rho_{21}^2-\rho_{2s}^2)P_2}{(1-\rho_{21}^2-\rho_{2s}^2)P_2+1}$. The above region is equivalent to \eqref{eq:inner211}, because it is obtained by substituting the equality constraint $P_2' = (1-\rho_{21}^2-\rho_{2s}^2)P_2$ into the two sum rate bounds in \eqref{eq:inner211} (which does not change the bounds), and relaxing the constraint on $P_2'$ to be $P_2' \leq (1-\rho_{21}^2-\rho_{2s}^2)P_2$, which affects only the first bound on $R_2$ but clearly does not enlarge the region. We now denote the bounds in \eqref{eq:inner211-1}-\eqref{eq:inner211-3} by $r_2(P_2')$, $r_{12}(\rho_{21},\rho_{2s})$, and $\tilde{r}_{12}(\rho_{21},\rho_{2s})$. For $0\leq P'_2 \leq P_2$, let $\left(\rho_{21}^*(P_2'),\rho_{2s}^*(P_2')\right)= \underset{(\rho_{21},\rho_{2s}):P_2'\leq (1-\rho_{21}^2-\rho_{2s}^2)P_2}{\text{argmax}} r_{12}(\rho_{21},\rho_{2s})$. Based on these notations, we characterize partial boundary of the capacity region for the Gaussian channel as follows.

\begin{figure}[thb]
\centering
\includegraphics[width=5in]{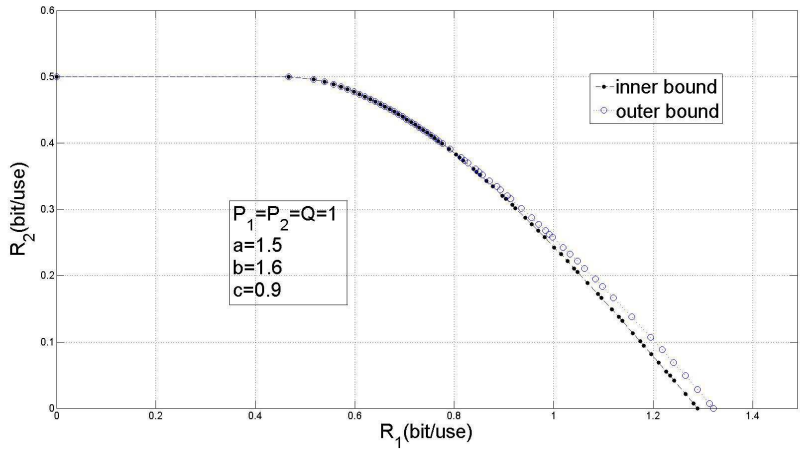}
\caption{An illustration of the partial boundary of the capacity region for a Gaussian channel with $|a| > 1$.}
\label{fig:GaussianCase2}
\end{figure}

\begin{theorem}(Partial Boundary of Capacity Region)\label{th:capa22}
Consider the Gaussian cognitive interference channel with state noncausally known at transmitter 2 and with $|a|> 1$. For $0\leq P'_2 \leq P_2$, the rate pairs $\Big(r_{12}\big(\rho^*_{21}(P_2'),\rho^*_{2s}(P_2')\big)-r_2(P_2'), \: r_2(P_2')\Big)$ is on the
boundary of the capacity region if $r_{12}(\rho^*_{21}(P_2'),\rho^*_{2s}(P_2'))\leq \tilde{r}_{12}(\rho^*_{21}(P_2'),\rho^*_{2s}(P_2'))$. The rate pairs $(R_1,r_2(P_2))$ are also on the boundary of the capacity region if $R_1 \leq \min\{r_{12},\tilde{r}_{12}\}|_{\rho_{21}=0,\rho_{2s}=0}-r_2(P_2) $.
\end{theorem}
\begin{proof}
The rate pairs given in the theorem are achievable due to the condition given in the theorem. They are also on the boundary of the outer bound given in Proposition \ref{th:outer21}, because $r_2$ and $r_{12}$ are the same as the bounds on $R_1$ and on $R_1+R_2$, respectively, and the chosen parameters $(\rho^*_{21}(P_2'),\rho^*_{2s}(P_2'))$ for each value of $P_2'$ guarantees that the rate pairs are on the boundary. The second statement is clear because when $P_2'=P_2$, $R_2$ achieves the maximum value, and hence any such rate pairs are on the boundary if they are achievable.
\end{proof}

In Fig.~\ref{fig:GaussianCase2}, we demonstrate the partial boundary of the capacity region characterized in Theorem \ref{th:capa22}. We consider the channel defined by the parameters $P_1=P_2=Q=1$,
$a=1.5$, $b= 1.6$ and $c=0.9$. We plot the boundaries of the inner bound given in Proposition \ref{th:inner21} and the outer bound given in Proposition \ref{th:outer21}, respectively. It is clear that the two boundaries match when $R_2$ is above a certain threshold, and this part is thus the boundary points of the capacity region characterized by Theorem \ref{th:capa22}.

We next show that under certain channel conditions, the outer bound given in Proposition \ref{th:outer21} fully characterizes the capacity region.
\begin{theorem}(Capacity)\label{th:capafull2}
For the Gaussian cognitive interference channel with state noncausally known at transmitter 2, if $|a|>1$ and the channel satisfies the condition \eqref{eq:cond4}, the capacity region consists of rate
pairs $(R_1,R_2)$ satisfying:
\begin{small}
\begin{flalign}
R_2 \leqslant & \frac{1}{2} \log(1+P_2')\nn\\
R_1+R_2 \leqslant & \frac{1}{2} \log\left(1+\frac{b^2P_1+2b\rho_{21}\sqrt{P_1P_2}+\rho_{21}^2P_2}{(1-\rho_{21}^2)P_2+2c\rho_{2s}\sqrt{P_2Q}+c^2Q+1}\right)+\frac{1}{2} \log(1+P_2')\;. \label{eq:capa21}
\end{flalign}
\end{small}
where $P_2'=(1-\rho_{21}^2-\rho_{2s}^2)P_2$ and $\rho_{21}^2+\rho_{2s}^2 \leqslant 1$.
\end{theorem}
\begin{proof}
Following from the region in \eqref{eq:DMCinner2}, and applying the condition \eqref{eq:cond4}, we obtain an inner
bound that includes the following bounds:
\begin{flalign}
R_2 \leqslant & I(U;Z|X_1) - I(U;S|X_1)\nn\\
R_1 + R_2 \leqslant & I(X_1U;Z) - I(U;S|X_1)
\end{flalign}
Based on the above bounds, we set the random variables as in \eqref{eq:innerbound2} and obtain an achievable region as given in \eqref{eq:capa21}. Such an achievable region is equivalent to the outer bound given in Proposition \ref{th:outer21} as we comment in the proof of Proposition \ref{th:outer21}.
\end{proof}

\subsection{Gaussian Channel: $|a|\leqslant 1$}\label{sec:gauss1}

\subsubsection{Inner and Outer Bounds}\label{sec:InOtC1}

%

We first note that the inner bound given in Proposition \ref{th:inner21} for the case when $|a|> 1$ also serves as an inner bound for the case when $|a|\leqslant 1$. However, the choice of auxiliary random variables ($T=\phi$ and $U=V$) for obtaining this inner bound requires receiver 1 to decode all information for receiver 2. As such, this bound works well only when receiver 1 is stronger than receiver 2, and does not serve as a good inner bound for the case when $|a|\leqslant 1$. Thus, in this subsection, we develop two new inner bounds and one new outer bound on the capacity region for the case when $|a|\leqslant 1$. We also note that the outer bound in Proposition \ref{th:outer21} is also applicable and useful here as demonstrated in the sequel.

The two inner bounds are derived based on the same achievable region for the discrete memoryless channel with different choices of the distributions for the auxiliary random variables. For the first inner bound, we design the dirty paper coding to deal with the state for receiver 1, and for the second inner bound, we design the dirty paper coding to deal with the state for receiver 2.
\begin{proposition}(Inner Bound 1)\label{th:inner1}
For the Gaussian cognitive interference channel with state noncausally known at transmitter 2, if $|a| \leqslant 1$, then an inner bound on the capacity region consists of rate pairs $(R_1, R_2)$
satisfying
\begin{flalign}
R_1 \leqslant & \frac{1}{2} \log\left(1+\frac{P_1+2a\rho_{21}\sqrt{P_1P_2}+a^2\rho_{21}^2P_2}{a^2(1-\rho_{21}^2)P_2+2a\rho_{2s}\sqrt{P_2Q}+Q+1}\right)+\frac{1}{2}\log\left(1+\frac{a^2 P_2'}{a^2 P_2''+1}\right) \nn\\
R_2 \leqslant &\frac{1}{2} \log(1+P_2'') \nn\\
R_2 \leqslant & \frac{1}{2}\log \left( 1+ \frac{a^2 {P_2'}^2 +2a\rho_{2s1} \rho_{2S2}P_2' -\rho_{2S1}^2(P_2'+P_2''+1) }{a^2P_2'P_2''+
 \rho_{2s1}^2(P_2'+P_2''+1)+a^2\rho_{2s2}P_2'+a^2P_2'-2a\rho_{2s1}\rho_{2s2}P_2'}\right) \nn\\
& + \frac{1}{2} \log(1+P_2'')\nn\\
R_1+R_2 \leqslant & \frac{1}{2} \log\left(1+\frac{b^2P_1+2b\rho_{21}\sqrt{P_1P_2}+\rho_{21}^2P_2}{(1-\rho_{21}^2)P_2+2c\rho_{2s}\sqrt{P_2Q}+c^2Q+1}\right) \nn\\
 &  + \frac{1}{2}\log \left( 1+ \frac{a^2 {P_2'}^2 +2a\rho_{2s1} \rho_{2s2}P_2' -\rho_{2s1}^2(P_2'+P_2''+1) }{a^2P_2'P_2''+
 \rho_{2s1}^2(P_2'+P_2''+1)+a^2\rho_{2s2}P_2'+a^2P_2'-2a\rho_{2s1}\rho_{2s2}P_2'}\right) \nn \\
 & + \frac{1}{2} \log(1+P_2'')
\end{flalign}
where $\rho_{2s1}= \alpha \left( 1+a \rho_{2s}\sqrt{\frac{P_2}{Q}}\right) \sqrt{Q}$, $\rho_{2s2}=\left(c+\rho_{2s}\sqrt{\frac{P_2}{Q}}\right)\sqrt{Q}$, $\alpha=\frac{a^2P_2'}{a^2P_2'+a^2P_2''+1}$, $|\rho_{21}| \leqslant 1$, $|\rho_{2s}| \leqslant 1$, $P_2'\geqslant 0$, $P_2''\geqslant 0$, and $P_2'+P_2''=(1-\rho_{21}^2-\rho_{2s}^2)P_2$.
\end{proposition}
\begin{proof}
The above theorem is based on Corollary \ref{cor:innerdm2} by choosing $(T,V,X_1,X_2)$ to be jointly Gaussian and employing dirty paper coding with $T$ chosen for dealing with the state for $Y$ and $V$ chosen for dealing with the state for $Z$. More specifically, We set the random variables as follows:
\begin{flalign}
      &X_1 \sim \mathcal{N}(0,P_1),\quad X_2' \sim \mathcal{N}(0,P_2'), \quad X_2'' \sim \mathcal{N}(0,P_2''), \quad P_2' + P_2'' = (1-\rho_{21}^2-\rho_{2s}^2)P_2\nn\\
      &X_2 = \rho_{21} \sqrt{\frac{P_2}{P_1}}X_1+ X_2'+X_2''+ \rho_{2s} \sqrt{\frac{P_2}{Q}}S\nn\\
      &T = X_2' + \alpha \left(1 + a\rho_{2s}\sqrt{\frac{P_2}{Q}}\right)S\nn\\
      &V = X_2'' + \beta \left(c - \alpha + (1-a\alpha)\rho_{2s}\sqrt{\frac{P_2}{Q}}\right)S \label{eq:dist11}
\end{flalign}
where $X_1$, $X_2'$, $X_2''$ and $S$ are independent random variables, $\alpha = \frac{a^2P_2'}{a^2P_2' + a^2P_2'' + 1}$, and $\beta = \frac{P_2''}{P_2'' + 1}$.
\end{proof}
\begin{proposition}(Inner Bound 2)\label{th:inner12}
For the Gaussian cognitive interference channel with state noncausally known at transmitter 2, if $|a| \leqslant 1$, then an inner bound on the capacity region consists of rate pairs $(R_1, R_2)$
satisfying
\begin{small}
\begin{flalign}
R_1 \leqslant & \frac{1}{2} \log\left(1+\frac{P_1+2a\rho_{21}\sqrt{P_1P_2}+a^2\rho_{21}^2P_2}{a^2(1-\rho_{21}^2)P_2+2a\rho_{2s}\sqrt{P_2Q}+Q+1}\right) \label{eq:inner121}\nn\\
& +\frac{1}{2}\log \left(1+\frac{a^2P_2'^2+2a\rho_{2s1}\rho_{2s2}P_2'-a^2\rho_{2s1}^2(P_2'+P_2'')-\rho_{2s1}^2}{a^2\rho_{2s1}^2P_2'+\rho_{2s2}^2P_2'+a^2\rho_{2s1}^2P_2''+a^2P_2'P_2''+P_2'+\rho_{2s1}^2-2a\rho_{2s1}\rho_{2s2}P_2'}\right)\\
R_2 \leqslant & \frac{1}{2} \log(1+P_2'') \label{eq:inner122} \\
R_1+R_2 \leqslant & \frac{1}{2} \log\left(1+\frac{b^2P_1+2b\rho_{21}\sqrt{P_1P_2}+\rho_{21}^2P_2}{(1-\rho_{21}^2)P_2+2c\rho_{2s}\sqrt{P_2Q}+c^2Q+1}\right) \nn \\
&+ \frac{1}{2} \log(1+(1-\rho_{21}^2-\rho_{2s}^2)P_2) \label{eq:inner123}
\end{flalign}
\end{small}
where $\rho_{2s1} = \alpha(c\sqrt{Q}+\rho_{2s}\sqrt{P_2})$, $\rho_{2s2} = (\sqrt{Q}+a\rho_{2s}\sqrt{P_2})$, $\alpha=\frac{P_2'}{P_2'+P_2''+1}$, $|\rho_{21}| \leqslant 1$, $|\rho_{2s}| \leqslant 1$, $P_2'\geqslant 0$, $P_2''\geqslant 0$, and $P_2'+P_2''=(1-\rho_{21}^2-\rho_{2s}^2)P_2$.
\end{proposition}
\begin{proof}
The above theorem is based on Corollary \ref{cor:innerdm2} by choosing $(T,V,X_1,X_2)$ to be jointly Gaussian and employing dirty paper coding by choosing $T$ and $V$ as follows:
  \begin{flalign}
      &X_1 \sim \mathcal{N}(0,P_1), \quad X_2' \sim \mathcal{N}(0,P_2'), \quad X_2'' \sim \mathcal{N}(0,P_2''), \quad P_2' + P_2'' = (1-\rho_{21}^2-\rho_{2S}^2)P_2 \nn\\
      &X_2 = \rho_{21} \sqrt{\frac{P_2}{P_1}}X_1+ X_2'+X_2''+ \rho_{2s} \sqrt{\frac{P_2}{Q}}S\nn\\
      &T = X_2' + \alpha \left(c + \rho_{2s}\sqrt{\frac{P_2}{Q}}\right)S\nn\\
      &V = X_2'' + \beta (1 - \alpha) \left(c + \rho_{2s}\sqrt{\frac{P_2}{Q}}\right)S
  \end{flalign}
where $X_1$, $X_2'$, $X_2''$ and $S$ are independent random variables, $\alpha = \frac{P_2'}{P_2' + P_2'' + 1}$, and $\beta = \frac{P_2''}{P_2'' + 1}$.
Here, $T$ is chosen for dealing with the state for $Z$ (different from the proof for Proposition \ref{th:inner1}) based on dirty paper coding where $X_2''$ is taken as noise. We then subtract $T$ from $Z$ and design $V$ for dealing with the state for $Z'=Z-T$ based on dirty paper coding. For this choice of the auxiliary random variables, the second bound on $R_2$ in Corollary \ref{cor:innerdm2} is redundant because $I(T;Z|X_1)-I(T;S|X_1)>0$.
\end{proof}

We next provide two outer bounds, both of which are useful for characterizing the capacity results in the following subsection. The first outer bound is given by the capacity region of the Gaussian interference channel with state known at both transmitter 2 and receiver 2 that we present as Theorem \ref{thr:GauWithSn1Con} in Section \ref{sec:recstate}. For convenience, we rewrite this bound below.
\begin{corollary}(Outer Bound 1)\label{cor:outer11}
For the Gaussian cognitive interference channel with state noncausally known at transmitter 2, if $|a| \leqslant 1$, then the capacity region of the same channel but with state known at both transmitter 2 and receiver 2 serves as an outer bound on the capacity region, which consists of rate pairs $(R_1, R_2)$ satisfying
\begin{small}
\begin{flalign}
R_1 \leqslant & \frac{1}{2} \log\left(1+\frac{P_1+2a\rho_{21}\sqrt{P_1P_2}+a^2\rho_{21}^2P_2}{a^2(1-\rho_{21}^2)P_2+2a\rho_{2s}\sqrt{P_2Q}+Q+1}\right)+\frac{1}{2}\log\left(1+\frac{a^2 P_2'}{a^2 P_2''+1}\right)\nn\\
R_2 \leqslant &\frac{1}{2} \log(1+P_2'')\nn\\
R_1+R_2 \leqslant & \frac{1}{2} \log(1+b^2P_1+2b\rho_{21}\sqrt{P_1P_2}+(1-\rho_{2s}^2)P_2)\nn
\end{flalign}
\end{small}
where $P_2'+P_2''=(1-\rho_{21}^2-\rho_{2s}^2)P_2$, $P_2'\geqslant 0$, $P_2''\geqslant 0$, and $\rho_{21}^2+\rho_{2s}^2 \leqslant 1$.
\end{corollary}

As we comment at the beginning of this subsection, the outer bound in Proposition \ref{th:outer21} is also applicable and useful for the case with $|a| \leqslant 1$. For convenience, we rewrite it below as a corollary.
\begin{corollary}\label{th:outer12}(Outer Bound 2)
For the Gaussian cognitive interference channel with state noncausally known at transmitter 2, if $|a| \leqslant 1$, an outer bound on the capacity region consists of rate pairs $(R_1,R_2)$
satisfying:
\begin{small}
\begin{flalign}
R_2 \leqslant & \frac{1}{2} \log(1+P_2'') \label{eq:outer121} \\
R_1+R_2 \leqslant & \frac{1}{2} \log\left(1+\frac{b^2P_1+2b\rho_{21}\sqrt{P_1P_2}+\rho_{21}^2P_2}{(1-\rho_{21}^2)P_2+2c\rho_{2s}\sqrt{P_2Q}+c^2Q+1}\right)+ \frac{1}{2} \log(1+(1-\rho_{21}^2-\rho_{2s}^2)P_2) \label{eq:outer122}
\end{flalign}
\end{small}
where $P_2''\leqslant (1-\rho_{21}^2-\rho_{2s}^2)P_2$, $P_2''\ge 0$, and $\rho_{21}^2+\rho_{2s}^2 \leqslant 1$.
\end{corollary}

\subsubsection{Capacity Theorems}

For Gaussian channels with $|a|\leqslant 1$, we characterize partial boundaries of the capacity region based on the inner and outer bounds in Section \ref{sec:InOtC1}. Although the forms of inner bounds are complicated, we show that some boundary points on the capacity region are determined only by a subset of there bounds, and can hence be characterized via the given outer bounds.

We let $\Delta=(\rho_{21}, \rho_{2s}, P_2')$ and use $r'_1(\Delta,P_2'')$, $r'_2(P_2'')$, $\tilde{r}'_2(\Delta,P_2'')$, $r'_{12}(\Delta,P_2'')$ to denote the four bounds on $R_1$, $R_2$, and $R_1+R_2$ given in Proposition \ref{th:inner1}. For $0 \leq P''_2 \leq P_2$, let $\Delta^*(P_2'')= \underset{\Delta:P_2'+P_2''= (1-\rho_{21}^2-\rho_{2s}^2)P_2}{\text{argmax}} r'_1(\Delta,P_2'')$. Based on these notations, we characterize partial boundary of the capacity region for the Gaussian channel as follows.
\begin{theorem}(Partial Boundary of Capacity Region)\label{th:capa11-1}
Consider the Gaussian cognitive interference channel with state noncausally known at transmitter 2 and with $|a|\leqslant 1$. For $0\leq P''_2 \leq P_2$, the rate pairs $(r'_1(\Delta^*(P_2''),P_2''),r'_2(P_2''))$ is on the boundary of the capacity region if $r'_2(P_2'')\leq \tilde{r}'_2(\Delta^*(P_2''),P_2'')$ and $r'_1(\Delta^*(P_2''),P_2'')+r'_2(P_2'') \leq r'_{12}(\Delta^*(P_2''),P_2'')$.
\end{theorem}
\begin{proof}
The rate pairs given in the theorem are contained in inner bound 1 given in Proposition \ref{th:inner1} due to the conditions given in the theorem. We next show that these rate pairs are also on the boundary of an outer bound. Following from outer bound 1 in Corollary \ref{cor:outer11}, $R_1\leqslant r'_1(\Delta,P_2'')$ and $R_2 \leqslant r'_2(P_2'')$ also determine an outer bound with $(\Delta,P_2'')$ taking the same values as in inner bound 1 given in Proposition \ref{th:inner1}. Then the chosen parameters $\Delta^*(P_2'')$ for each value of $P_2''$ guarantees that the rate pairs are on the boundary of this outer bound.
\end{proof}
\begin{remark}
The rate pairs characterized in Theorem \ref{th:capa11-1} are on the boundary of the capacity region with the state known at both transmitter 2 and receiver 2, which is the outer bound 1 in Corollary \ref{cor:outer11}.
\end{remark}

We next characterize additional boundary points of the capacity region based on inner bound 2 given in Proposition \ref{th:inner12} and outer bound 2 given in Corollary \ref{th:outer12}. We use $r''_1(\rho_{21}, \rho_{2s},P_2',P_2'')$, $r''_2(P_2'')$, and $r''_{12}(\rho_{21}, \rho_{2s})$ to denote the three bounds on $R_1$, $R_2$, and $R_1+R_2$ in inner bound 2 given in Proposition \ref{th:inner12}. For $0 \leq P''_2 \leq P_2$, let $(\rho_{21}^*(P_2''), \rho_{2s}^*(P_2''))= \underset{(\rho_{21}, \rho_{2s}):P_2''\leq (1-\rho_{21}^2-\rho_{2s}^2)P_2}{\text{argmax}} r''_{12}(\rho_{21}, \rho_{2s})$, and let $P_2'^*(P_2'')=(1-{\rho_{21}^*(P_2'')}^2-{\rho_{2s}^*(P_2'')}^2)P_2-P_2''$. Based on these notations, we characterize partial boundary of the capacity region as follows.
\begin{theorem}(Partial Boundary of Capacity Region)\label{th:capa11-2}
Consider the Gaussian cognitive interference channel with state noncausally known at transmitter 2 and with $|a| \leqslant 1$. For $0 \leq P''_2 \leq P_2$, the rate pairs $(r''_{12}(\rho_{21}^*(P_2''), \rho_{2s}^*(P_2''))-r''_2(P_2''),\: r''_2(P_2''))$ is on the boundary of the capacity region if $r''_{12}(\rho_{21}^*(P_2''), \rho_{2s}^*(P_2''))-r''_2(P_2'') \leq r''_1(\rho_{21}^*(P_2''), \rho_{2s}^*(P_2''),P_2'^*(P_2''),P_2'')$. The rate pairs $(R_1,r''_2(P_2))$ are also on the boundary of the capacity region if $R_1 \leq \min\{r''_1,r''_{12}-r''_2(P_2)\}|_{\rho_{21}=0,\rho_{2s}=0,P_2'=0}$.
\end{theorem}
\begin{proof}
The rate pairs given in the theorem are clearly contained in inner bound 2 given in Proposition \ref{th:inner12}. These rate pairs are also on the boundary of outer bound 2 given in Corollary \ref{th:outer12}, because $r''_2$ and $r''_{12}$ are the same as the bounds on $R_2$ and on $R_1+R_2$, respectively, and the chosen parameters $(\rho^*_{21}(P_2''),\rho^*_{2s}(P_2''))$ for each value of $P_2''$ guarantees that the rate pairs are on the boundary. The second statement is clear because when $P_2''=P_2$, $R_2$ achieves the maximum value, and hence any rate pairs with such $R_2$ are on the boundary if they are achievable.
\end{proof}

Theorems \ref{th:capa11-1} and \ref{th:capa11-2} collectively characterize partial boundary of the capacity region for the Gaussian channel with $|a| \leqslant 1$. In Fig.~\ref{fig:GaussianCase1}, we demonstrate these boundary points of the capacity region for an example channel with the parameters $P_1=P_2=Q=1$, $b=0.85$, $c=0.9$ and $a=0.8$. We plot the boundaries of the two inner bounds given in Proposition \ref{th:inner1} and Proposition \ref{th:inner12}, and the boundaries of the two outer bounds given in Corollary \ref{cor:outer11} and Corollary \ref{th:outer12}, respectively. We observe that the two inner bounds are very close. It can be seen that the boundary of inner bound 1 matches the boundary of outer bound 1 when $R_1$ is above a certain value, and this part is thus on the boundary of the capacity region. We also note that this part of the boundary achieves the capacity region of the same channel with state also known at receiver 2. It can further be seen that the boundary of inner bound 2 matches the boundary of outer bound 2 when $R_2$ is above a certain threshold, and this part is hence also on the boundary of the capacity region.

\begin{figure}[thb]
\centering
\includegraphics[width=5in]{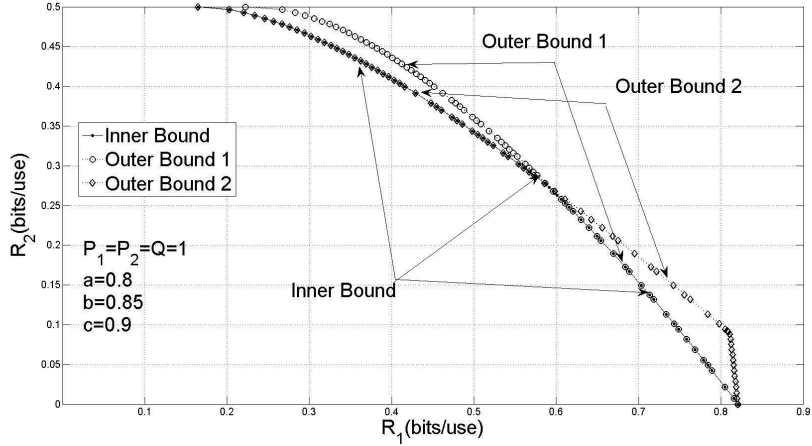}
\caption{An illustration of inner and outer bounds and the partial boundary of the capacity region for a Gaussian channel with $|a| \leqslant 1$}
\label{fig:GaussianCase1}
\end{figure}


It can be seen that outer bounds 1 and 2 separately characterize certain parts of the boundary of the capacity region for Gaussian channels with $|a| \leqslant 1$. We further show that each of these two outer bounds can characterize the full capacity region for channels that satisfy certain conditions.
\begin{theorem}(Capacity)\label{th:capa13}
For the Gaussian cognitive interference channel with state noncausally known at transmitter 2, if $|a| \leqslant 1$ and the channel satisfies the condition (\ref{eq:cond3}), the capacity region consists of rate pairs $(R_1,R_2)$ satisfying:
\begin{small}
\begin{flalign}
R_1 \leqslant & \frac{1}{2} \log\left(1+\frac{P_1+2a\rho_{21}\sqrt{P_1P_2}+a^2\rho_{21}^2P_2}{a^2(1-\rho_{21}^2)P_2+2a\rho_{2s}\sqrt{P_2Q}+Q+1}\right)+\frac{1}{2}\log\left(1+\frac{a^2 P_2'}{a^2 P_2''+1}\right) \nn \\
R_2 \leqslant &\frac{1}{2} \log(P_2''+1) \label{eq:capa13}
\end{flalign}
\end{small}
where $P_2'+P_2''=(1-\rho_{21}^2-\rho_{2S}^2)P_2$, $P_2' \geqslant 0$, $P_2''\geqslant 0$ and $\rho_{21}^2+\rho_{2S}^2 \leqslant 1$, $|\rho_{21}| \leqslant 1$, $|\rho_{2S}| \leqslant 1$.
\end{theorem}
\begin{proof}
Under the condition (\ref{eq:cond3}), the bounds in the achievable region in Corollary \ref{cor:innerdm2} reduce to:
\begin{flalign}
R_1 \leqslant & I(X_1T;Y) - I(T;S|X_1)\nn\\
R_2 \leqslant & I(V;Z|X_1T) - I(V;S|X_1T)\nn\\
R_2 \leqslant & I(TV;Z|X_1) - I(TV;S|X_1)\label{eq:capa13DMC}
\end{flalign}
Based on the above bounds, we choose the same jointly Gaussian input distribution as in \eqref{eq:dist11}. In particular, since the auxiliary random variable $T$ is chosen to employ dirty paper coding to deal with the state in $Y$, it guarantees that $I(T;Y|X_1)-I(T;S|X_1) \geqslant 0$, which implies that $I(T;Z|X_1)-I(T;S|X_1) \geqslant 0$ due to the condition \eqref{eq:cond3}. Hence, the third bound in \eqref{eq:capa13DMC} is redundant. Thus, we obtain an achievable region that matches the first two bounds of outer bound 1 in Corollary \ref{cor:outer11} and is hence tight.
\end{proof}
\begin{remark}
The above theorem implies that the Gaussian cognitive interference channel with state noncausally known only to transmitter 2 achieves the capacity of the same channel with state also known to receiver 2 if the channel satisfies $|a| \leq 1$ and the condition (\ref{eq:cond3}). This is similar to the result that dirty paper coding achieves the capacity of the Gaussian channel with state also known at the receiver \cite{Costa83}. Here, the channel cannot achieve the capacity with both receivers knowing the channel state due to the fact that transmitter 1 does not know the channel state.
\end{remark}

We note that the above region matches the capacity in \cite{Somekh08} of another cognitive interference model with state, in which $W_1$  is intended only for receiver 1. This is reasonable because under the condition \eqref{eq:cond3}, receiver 1 is weaker in decoding $W_1$ than receiver $2$, and receiver 2 can hence always decode $W_1$, which satisfies the additional requirement in the model of this paper.


The following theorem identifies the channels for which outer bound 2 given in Corollary \ref{th:outer12} characterizes the full capacity region.
\begin{theorem}(Capacity)\label{th:capa12}
For the Gaussian cognitive interference channel with state noncausally known at transmitter 2, if $|a| \leqslant 1$ and the channel satisfies the condition (\ref{eq:cond4}), the capacity region consists
of rate pairs $(R_1,R_2)$ satisfying:
\begin{small}
\begin{flalign}
R_1 \leqslant &\frac{1}{2}\log \left(1+\frac{P_2'}{ P_2''+1}\right)+\frac{1}{2} \log\left(1+\frac{b^2P_1+2b\rho_{21}\sqrt{P_1P_2}+\rho_{21}^2P_2}{(1-\rho_{21}^2)P_2+2c\rho_{2s}\sqrt{P_2Q}+c^2Q+1}\right)\nn\\
R_2 \leqslant &  \frac{1}{2} \log(1+P_2'') \label{eq:capa12}
\end{flalign}
\end{small}
where $P_2'+P_2''=(1-\rho_{21}^2-\rho_{2S}^2)P_2$, $P_2' \geqslant 0$, $P_2''\ge 0$ and $\rho_{21}^2+\rho_{2S}^2 \leqslant 1$, $|\rho_{21}| \leqslant 1$, $|\rho_{2S}| \leqslant 1$.
\end{theorem}
\begin{proof}
With the condition \eqref{eq:cond4}, it can be seen that an achievable region determined by the following bounds is contained in the inner bound given in Corollary \ref{cor:innerdm2}, and is hence achievable.
\begin{flalign}
R_1 \leqslant & I(X_1T;Z) - I(T;S|X_1)\nn\\
R_2 \leqslant & I(V;Z|X_1T) - I(V;S|X_1T)\nn\\
R_2 \leqslant & I(TV;Z|X_1) - I(TV;S|X_1)\;.\label{eq:ach12DMC}
    \end{flalign}
The achievability follows from the above region by choosing the jointly Gaussian distribution and employing dirty paper coding for $T$ to deal with the state for $Z$ and for $V$ to deal with the remaining state for $Z$ after subtracting $\frac{1}{a}T$. More specifically, we set the auxiliary random variable as follows:
\begin{flalign}
      &X_1 \sim \mathcal{N}(0,P_1), \quad X_2' \sim \mathcal{N}(0,P_2'), \quad X_2'' \sim \mathcal{N}(0,P_2''), \quad P_2' + P_2'' = (1-\rho_{21}^2-\rho_{2S}^2)\nn\\
      &X_2 = \rho_{21} \sqrt{\frac{P_2}{P_1}}X_1+ X_2'+X_2''+ \rho_{2s} \sqrt{\frac{P_2}{Q}}S\nn\\
      &T = X_2' + \alpha \left(c + \rho_{2s}\sqrt{\frac{P_2}{Q}}\right)S\nn\\
      &V = X_2'' + \beta (1 - \alpha) \left(c + \rho_{2s}\sqrt{\frac{P_2}{Q}}\right)S
\end{flalign}
where $X_1$, $X_2'$, $X_2''$ and $S$ are independent random variables, $\alpha = \frac{P_2'}{P_2' + P_2'' + 1}$, and $\beta = \frac{P_2''}{P_2'' + 1}$. Such a choice of the input distribution also implies that $I(T;Z|X_1) - I(T;S|X_1) \geqslant 0$, and the third bound in \eqref{eq:ach12DMC} is hence redundant. The proof for the converse follows by observing that the region \eqref{eq:capa12} has the same boundary points as outer bound 2 given in Corollary \ref{th:outer12}, and hence the two regions are equivalent.
\end{proof}
We note that Theorems \ref{th:capafull2} and \ref{th:capa12} implies that under the condition (\ref{eq:cond4}), the Gaussian cognitive interference channel with state has the same capacity region as the multiple access channel with state given in \cite{MACSomekhBaruch08}. This is reasonable because the condition (\ref{eq:cond4}) implies that receiver 2 is weaker than receiver 1 in decoding $W_1$, and hence dominates the rate region.

\section{State Known at both Transmitter 2 and Receiver 2}\label{sec:recstate}

In this section, we study the cognitive interference channel with state known at both transmitter 2 and receiver 2. This channel is of interest by its own, and the capacity of this channel also provides a useful outer bound for characterizing the capacity for the channel with the state known only at transmitter 2 as already demonstrated in Section \ref{sec:gauss1}.

\subsection{Discrete Memoryless Channel}

We characterize the full capacity region in the following theorem. In particular, the proof of the converse applies the techniques developed recently in \cite{Lapidoth11} for proving equivalence of two regions characterized by different sets of auxiliary random variables.
\begin{theorem}\label{thr:WithSn}(Capacity)
The capacity region for the cognitive interference channel with state noncausally known at both transmitter 2 and receiver 2 consists of rate pairs $(R_1, R_2)$ satisfying:
\begin{flalign}
R_1 \leqslant & I(X_1U;Y) - I(U;S|X_1)\nn\\
R_2 \leqslant & I(X_2;Z|S X_1)\nn\\
R_1 + R_2 \leqslant & I(X_1X_2;Z|S)\nn\\
R_1 + R_2 \leqslant & I(X_1U;Y) + I(X_2;Z|X_1US) - I(U;S|X_1)\label{eq:WithSnEq}
\end{flalign}
for some distribution  $P_{X_1SUX_2YZ} = P_{X_1}P_{S}P_{UX_2|X_1S}P_{YZ|SX_1X_2}$, where $U$ is an auxiliary random variable and its cardinality is bounded by $|\cU|\leqslant |\cX_1||\cX_2||\cS|+1$.
\end{theorem}
\begin{proof}
The achievability follows from the achievable region given in \eqref{OrigAch} by setting $T=X_1$, $V=X_2$ and $Z = ZS$.

For the converse, we first obtain the following outer bound consisting of rate pairs $(R_1,R_2)$ satisfying
\begin{flalign}
    R_1 & \leqslant I(K X_1;Y) - I(K;S|X_1)\nn\\
    R_2 & \leqslant I(X_2;Z|S X_1)\nn\\
    R_1 + R_2 & \leqslant I(X_1  X_2;Z|S )\nn\\
    R_1 + R_2 & \leqslant I(TKX_1;Y) - I(T K;S|X_1) + I (X_2;Z|X_1TKS) \label{eq:outerbothstate}
\end{flalign}
for some distribution  $P_{X_1STKX_2YZ} = P_{X_1}P_{S}P_{KT|X_1S}P_{X_2|X_1SKT}P_{YZ|SX_1X_2}$, where $K$ and $T$ are auxiliary random variables. The proof is detailed in Appendix \ref{pr:ConWithSn}.

In order to show that the region \eqref{eq:WithSnEq} is the capacity region, it is sufficient to show that the above outer bound \eqref{eq:outerbothstate} is a subset of the region \eqref{eq:WithSnEq}. Towards this end, we apply the technique in \cite{Lapidoth11} and analyze the outer bound \eqref{eq:outerbothstate} by considering the following two cases.

%
%

If $I(T;Y|K X_1) - I(T;S|K X_1) \leqslant 0$, the outer bound \eqref{eq:outerbothstate} can be further bounded as:
\begin{flalign}
R_1 \leqslant & I(K X_1;Y) - I(K;S|X_1) \nn \\
R_2 \leqslant & I(X_2;Z|S X_1) \nn\\
R_1 + R_2 \leqslant & I(X_1 X_2;Z|S) \nn\\
R_1 + R_2\leqslant & I(K X_1;Y) - I(K;S|X_1) + [I(T;Y|K X_1) - I(T;S|K X_1)] + I (X_2;Z|X_1 T K S) \nn\\
   \leqslant & I(K X_1;Y) - I(K;S|X_1) + I (X_2;Z|X_1 K S)\label{eq:OuterWiSn}.
\end{flalign}
which implies that the outer bound \eqref{eq:outerbothstate} is contained in \eqref{eq:WithSnEq} by setting $U=K$ in \eqref{eq:WithSnEq}.

If $I(T;Y|K X_1) - I(T;S|K X_1) \geqslant 0$, the outer bound \eqref{eq:outerbothstate} can be further bounded as:
\begin{flalign}
R_1 \leqslant & I(K X_1;Y) - I(K;S|X_1) \nn\\
= & I(KT X_1; Y) - I(KT; S|X_1) - [I(T;Y|K X_1) - I(T;S|K X_1)]\nn \\
\leqslant &I(KT X_1; Y) - I(KT; S|X_1) \nn \\
    R_2 \leqslant & I(X_2;Z|S X_1) \nn\\
R_1 + R_2 \leqslant & I(X_1 X_2;Z|S) \nn\\
R_1 + R_2 \leqslant &  I(T K X_1;Y) - I(T K;S|X_1) + I (X_2;Z|X_1 K T S)
\end{flalign}
which also implies that the outer bound \eqref{eq:outerbothstate} is contained in \eqref{eq:WithSnEq} by setting $U=KT$ in \eqref{eq:WithSnEq}.
\end{proof}
\begin{remark}
By setting $X_1$ to be deterministic, Theorem \ref{thr:WithSn} reduces to the capacity region for the broadcast channel with degraded message sets and with state information noncausally known at both transmitter 2 and receiver 2, which consists of rate pairs satisfying
\begin{flalign}
R_1 \leqslant & I(U;Y) - I(U;S)\nn\\
R_1 + R_2 \leqslant & I(X;Z|S) \nn\\
R_1 + R_2 \leqslant & I(U;Y) + I(X;Z|US) - I(U;S)
\end{flalign}
for some distribution $P_{SUXYZ} = P_{S}P_{UX|S}P_{YZ|SX}$, where $X$ is the channel input, $Y$ and $Z$ are channel outputs respectively at two receivers.
\end{remark}

\subsection{Gaussian Channels}

In this section, we characterize the capacity region for Gaussian cognitive interference channels with state known at both transmitter 2 and receiver 2. The channel input-output relationship is the same as described in \eqref{eq:gaussmodel}. Similar to Section \ref{sec:generalGaussian}, we will partition Gaussian channels into two classes based on the value of the channel parameter $a$, and characterize the capacity region for each class.

We first provide the capacity region for the Gaussian channel with $|a| \leqslant 1$. As shown in Theorems \ref{th:capa11-1} and \ref{th:capa13}, this capacity region serves as the tight converse for characterizing the partial or full boundary of the capacity region when the state is known noncausally only at transmitter 2.
\begin{theorem}\label{thr:GauWithSn1Con}(Capacity)
For the Gaussian cognitive interference channel with state known at transmitter 2 and receiver 2, if $|a| \leqslant 1$, the capacity region consists of
rate pairs $(R_1,R_2)$  satisfying:
  \begin{small}
\begin{flalign}
R_1 \leqslant & \frac{1}{2} \log\left(1+\frac{P_1+2a\rho_{21}\sqrt{P_1P_2}+a^2\rho_{21}^2P_2}{a^2(1-\rho_{21}^2)P_2+2a\rho_{2s}\sqrt{P_2Q}+Q+1}\right)+\frac{1}{2}\log\left(1+\frac{a^2 P_2'}{a^2 P_2''+1}\right)\nn\\
R_2 \leqslant &\frac{1}{2} \log(1+P_2'')\nn\\
R_1+R_2 \leqslant & \frac{1}{2} \log\left(1+b^2P_1+2b\rho_{21}\sqrt{P_1P_2}+(1-\rho_{2s}^2)P_2\right)\label{eq:capawithSn3}
  \end{flalign}
  \end{small}
where $P_2' +P_2''= (1-\rho_{21}^2-\rho_{2s}^2)P_2$, $P_2' \geqslant 0$, $P_2'' \geqslant 0$, and $\rho_{21}^2+\rho_{2s}^2 \leqslant 1$.
\end{theorem}
\begin{proof}

Consider the following rate region, which consists of rate pairs $(R_1,R_2)$ satisfying
\begin{flalign}
R_1 \leqslant & I(X_1U;Y) - I(U;S|X_1)\nn\\
R_2 \leqslant & I(X_2;Z|UX_1S)\nn \\
R_1+R_2 \leqslant & I(X_1X_2;Z|S)\label{eq:DegrConvWithSn}
\end{flalign}
for some distribution $P_{SX_1UX_2YZ} =P_{X_1}P_{S}P_{UX_2|X_1S}P_{Z|X_1X_2S}P_{Y|ZX_1S}$. This region is contained in \eqref{eq:WithSnEq}, and is hence achievable. This can be seen by observing that $I(X_2;Z|UX_1S) \leqslant I(X_2U;Z|X_1S)$ and the second sum rate bound in \eqref{eq:WithSnEq} is equal to the sum of the two bounds on the individual rates in \eqref{eq:DegrConvWithSn}.




The achievability of \eqref{eq:capawithSn3} is then obtained by choosing the following jointly Gaussian distribution for the random variables:
\begin{flalign}
      &X_1 \sim \mathcal{N}(0,P_1),\quad X_2' \sim \mathcal{N}(0,P_2'), \quad X_2'' \sim \mathcal{N}(0,P_2''), \quad P_2' + P_2'' = (1-\rho_{21}^2-\rho_{2s}^2)P_2\nn\\
      &X_2 = \rho_{21} \sqrt{\frac{P_2}{P_1}}X_1+ X_2'+X_2''+ \rho_{2s} \sqrt{\frac{P_2}{Q}}S\nn\\
      &U = X_2' + \alpha \left(1 + a\rho_{2s}\sqrt{\frac{P_2}{Q}}\right)S
\end{flalign}
where $X_1$, $X_2'$ , $X_2''$ and $S$ are independent, and $\alpha = \frac{a^2P_2'}{a^2P_2' + a^2P_2'' + 1}$. Here, the auxiliary random variable $U$ is designed based on dirty paper coding to deal with the state for receiver 1.

The converse proof is detailed in Appendix \ref{apx:ProofGauWithSn1Con}.
\end{proof}

We next characterize the capacity region for the Gaussian channel with $|a| > 1$.
\begin{theorem}\label{thr:GauWithSn2Con}(Capacity)
  For the Gaussian cognitive interference channel with state noncausally known at transmitter 2 and receiver 2, if $|a| > 1$, the capacity region consists of rate pairs $(R_1,R_2)$ satisfying:
\begin{small}
  \begin{flalign}
    R_2 \leqslant &\frac{1}{2} \log(1+(1-\rho_{21}^2-\rho_{2s}^2)P_2)\nn\\
    R_1+R_2 \leqslant &\frac{1}{2} \log(1+b^2P_1+2b\rho_{21}\sqrt{P_1P_2}+(1-\rho_{2s}^2)P_2) \nn\\
    R_1+R_2 \leqslant &\frac{1}{2} \log\left(1+\frac{P_1+2a\rho_{21}\sqrt{P_1P_2}+a^2\rho_{21}^2P_2}{a^2(1-\rho_{21}^2)P_2+2a\rho_{2s}\sqrt{P_2Q}+Q+1}\right)+\frac{1}{2}\log(1+a^2 (1-\rho_{2s}^2-\rho_{21}^2) P_2)
  \end{flalign}
  \end{small}
where $\rho_{21}^2+\rho_{2s}^2 \leqslant 1$.
\end{theorem}
\begin{proof}
  The achievability follows from \eqref{eq:WithSnEq} by choosing jointly Gaussian distribution for random variables as follows:
    \begin{flalign}
      &X_1 \sim \mathcal{N}(0,P_1),\quad X_2' \sim \mathcal{N}(0, (1-\rho_{21}^2-\rho_{2s}^2)P_2)\nn\\
      &X_2 = \rho_{21} \sqrt{\frac{P_2}{P_1}}X_1+ X_2'+ \rho_{2s} \sqrt{\frac{P_2}{Q}}S\nn\\
      &U = X_2' + \alpha \left(1 + a\rho_{2s}\sqrt{\frac{P_2}{Q}}\right)S
    \end{flalign}
where $X_1$, $X_2'$ and $S$ are independent, and $\alpha =\frac{a^2(1-\rho_{21}^2 - \rho_{2s}^2)P_2}{a^2(1-\rho_{21}^2 - \rho_{2s}^2)P_2+1}$. We note that with this choice of the random variables, the first bound in \eqref{eq:WithSnEq} is redundant.

In order to prove the converse for Theorem \ref{thr:GauWithSn2Con}, we first prove the following outer bound.
\begin{lemma}\label{lm:DegrConvWithSn2}
For the cognitive interference channel with state noncausally known at both transmitter 2 and receiver 2, if it satisfies the condition \eqref{eq:cond2}, an outer bound on the capacity region consists of rate pairs $(R_1,R_2)$ satisfying
\begin{flalign}
    R_2 & \leqslant I(X_2;Z|S X_1)\nn\\
    R_1 + R_2 & \leqslant I(X_1 X_2 ;Z|S) \nn\\
    R_1 + R_2 & \leqslant  I(X_1;Y) + I(X_2;Y|S X_1)\label{eq:outerWithSn2}
\end{flalign}
for some distribution $P_{SX_1UX_2YZ} =P_{X_1}P_{S}P_{UX_2|X_1S}P_{Y|X_1X_2S}P_{Z|YX_1S}$.
\end{lemma}

The proof for the above lemma is detailed in Appendix \ref{apx:ProofDMCGauWithSn2}. For the Gaussian channel with $|a| > 1$, it satisfies the condition \eqref{eq:cond2}. We then use the above lemma for developing the converse proof, which is detailed in Appendix \ref{apx:ProofGauWithSn2Con}.
\end{proof}

\section{Conclusion}\label{sec:conclusion}

In this paper, we performed a comprehensive study of a class of cognitive interference channels, which are corrupted by i.i.d. state sequences. We studied two cases, in which the state sequence is assumed to be noncausally known at transmitter 2, and at both transmitter 2 and receiver 2, respectively. We characterize inner and outer bounds on the capacity region for discrete memoryless and Gaussian channels. Our inner bounds are based on the Gel'fand-Pinsker scheme/dirty paper coding, rate splitting and superposition coding. Our outer bounds are constructed to match the inner bounds as much as possible. Based on these inner and outer bounds, we characterized the partial or full capacity regions for various channels.

In particular, we anticipate that our technique of characterizing partial boundary of the capacity region for the Gaussian channel may be applicable for other Gaussian network models. Furthermore, our characterization of the capacity region for the case with state known at both transmitter 2 and receiver 2 applies the technique developed recently in \cite{Lapidoth11} for proving equivalence of inner and outer bounds characterized by different sets of auxiliary random variables. Such a technique may also be useful for other network models.

\vspace{10mm}

\appendix

\noindent {\Large \textbf{Appendix}}

\section{Proof of Lemma \ref{NaiAchv}}\label{ProofNaiAch}
The achievable scheme applies rate splitting, superposition coding and Gel'fand-Pinsker binning scheme.
We use random codes and fix the following joint distribution:
\[P_{SX_1TUVX_2YZ}= P_{X_1}P_SP_{T|X_1S}P_{U|X_1TS}P_{V|TUX_1S}P_{X_2|TUVX_1S}P_{YZ|X_1X_2S}.\]
Let $T_\epsilon^n(P_{SX_1TUVX_2YZ})$ denote the strongly joint $\epsilon$-typical set based on the above distribution. For a given sequence $x^n$, let $T_\epsilon^n(P_{U|X}|x^n)$ denote the set of sequences $u^n$ such that $(u^n, x^n)$ is jointly typical based on the distribution $P_{XU}$.

\textit{Code Construction:}
\begin{list}{$$}{\topsep=0.ex \leftmargin=1cm
\rightmargin=0.in \itemsep =0.in}
\item[1.] Generate $2^{nR_1}$ codewords $x_1^n(w_1)$ with i.i.d.\ components based on $P_{X_1}$. Index these codewords by $w_1 = 1, \dotsi, 2^{nR_1}$. 

\item[2.] For each $x_1^n(w_1)$, generate $t^n(w_1,v_1)$ with i.i.d.\ components based on $P_{T|X_1}$. Index these codewords by $v_1=1, \dotsi , 2^{n\tilde{R}_1}$.

\item[3.] For each $x_1^n(w_1)$ and $t^n(w_1,v_1)$, generate $u^n(w_1, v_1, w_{21}, v_{21})$ with i.i.d.\ components based on $P_{U|X_1T}$. Index these codewords by $w_{21} = 1,\dotsi, 2^{nR_{21}}$ and $v_{21} = 1, \dotsi, 2^{n\tilde{R}_{21}}$.

\item[4.] For each $x_1^n(w_1)$, $t^n(w_1,v_1)$, and $u^n(w_1, v_1, w_{21}, v_{21})$, generate $v^n(w_1, v_1, w_{21}, v_{21}, w_{22}, v_{22})$ with i.i.d.\ components based on $P_{V|X_1 TU}$. Index these codewords by $w_{22} = 1, \dotsi, 2^{nR_{22}}$ and $v_{22} = 1, \dotsi, 2^{n\tilde{R}_{22}}$.

\end{list}

\textit{Encoding:}
\begin{list}{$$}{\topsep=0.ex \leftmargin=1cm
\rightmargin=0.in \itemsep =0.in}

\item[1.] Encoder 1: Given $w_1$, map $w_1$ into $x_1^n(w_1)$ for transmission.
\item[2.] Encoder 2:
\begin{list}{$-$}{\topsep=0.ex \leftmargin=0.22in
\rightmargin=0.in \itemsep =0.in}

\item Given $w_1$, $x_1^n(w_1)$ and $s^n$, select $t^n(w_1, \tilde{v}_1)$ such that
\begin{equation*}
  (t^n(w_1, \tilde{v}_1), s^n, x_1^n(w_1)) \in T_{\epsilon}^n(P_{X_1}P_SP_{T|X_1S})\textrm{.}
\end{equation*}

 Otherwise, set $\tilde{v}_1=1$. It can be shown that for large $n$, such $t^n$ exists with high probability if
\begin{equation}\label{eq:OriAch1}
  \tilde{R_1} > I(T;S|X_1)\textrm{.}
\end{equation}

\item Given $w_{21}$ and selected $t^n(w_1, \tilde{v}_1)$, select $u^n(w_1, \tilde{v}_1, w_{21}, \tilde{v}_{21})$ such that

\begin{equation*}
  (u^n(w_1, \tilde{v}_1, w_{21}, \tilde{v}_{21}), t^n(w_1, \tilde{v}_1), s^n, x_1^n(w_1)) \in T_{\epsilon}^n(P_{X_1}P_SP_{T|X_1S}P_{U|X_1ST})\textrm{.}
\end{equation*}

Otherwise, set $\tilde{v}_{21} = 1$. It can be shown that for large $n$, such $u^n$ exists with high probability if
\begin{equation}\label{eq:OriAch2}
  \tilde{R}_{21} > I(U;S|X_1T)\textrm{.}
\end{equation}

\item Given $w_{22}$ and selected $u^n(w_1, \tilde{v}_1, w_{21}, \tilde{v}_{21})$, select $v^n(w_1, \tilde{v}_1, w_{21}, \tilde{v}_{21}, w_{22}, \tilde{v}_{22})$ such that
    \begin{equation*}
    \begin{split}
      & (v^n(w_1, \tilde{v}_1, w_{21}, \tilde{v}_{21}, w_{22}, \tilde{v}_{22}),u^n(w_1,\tilde{v}_1, w_{21}, \tilde{v}_{21}), t^n(w_1, \tilde{v}_1), s^n, x_1^n(w_1)) \\
      & \in T_{\epsilon}^n(P_{X_1}P_SP_{T|X_1S}P_{U|X_1ST}P_{V|UX_1ST})\textrm{.}
    \end{split}
    \end{equation*}

    Otherwise, set $\tilde{v}_{22}=1$. It can be shown that for large $n$, such $v^n$ exists with high probability if
    \begin{equation}
      \tilde{R}_{22} > I(V;S|UX_{1}T)\label{eq:OriAch3}\textrm{.}
    \end{equation}
\item Given selected $x_1^n(w_1)$, $t^n(w_1, \tilde{v}_{1})$, $u^n(w_1, \tilde{v}_{1}, w_{21}, \tilde{v}_{21})$, $v^n(w_1, \tilde{v}_{1}, w_{21}, \tilde{v}_{21}, w_{22}, \tilde{v}_{22})$ and $s^n$, generate $x_2^n$ with i.i.d.\ components based on $P_{X_2|TUVX_1S}$ for transmission. 
\end{list}
\end{list}


\textit{Decoding:}
\begin{list}{$$}{\topsep=0.ex \leftmargin=1cm
\rightmargin=0.in \itemsep =0.in}
\item[1] Decoder 1: Given $y^n$, find the unique tuple $(\hat{w}_1, \hat{v_1}, \hat{w}_{21}, \hat{v}_{21})$ such that
  \begin{equation*}
    (x_1^n(\hat{w}_1), t^n(\hat{w}_1, \hat{v}_1), u^n(\hat{w}_1, \hat{v}_1 , \hat{w}_{21}, \hat{v}_{21}), y^n) \in T_\epsilon^n(P_{X_1TUY}).
  \end{equation*}

If no or more than one such tuples with different $w_1$ can be found, then declare error. One can show that for sufficiently large $n$, decoding is correct with high probability if
    \begin{flalign}
      &R_1 + \tilde{R_1} + R_{21} + \tilde{R}_{21} \leqslant I(TUX_1;Y)\label{eq:OriAch4}
    \end{flalign}

    We note that since receiver 1 is not required to decode $W_{21}$ correctly by the channel model, the corresponding error events do not need to be analyzed.

\item[2.] Decoder 2: Given $z^n$, find a tuple $(\hat{w}_1, \hat{v_1}, \hat{w}_{21}, \hat{v}_{21}, \hat{w}_{22}, \hat{v}_{22})$ such that
\begin{flalign}
  (x_1^n(\hat{w}_1), t^n(\hat{w}_1, \hat{v}_1), u^n&(\hat{w}_1, \hat{v}_1, \hat{w}_{21}, \hat{v}_{21}), v^n(\hat{w}_1, \hat{v}_1, \hat{w}_{21}, \hat{v}_{21}, \hat{w}_{22}, \hat{v}_{22}), z^n) \nn\\
  &\in T_\epsilon^n(P_{X_1TUVZ})\textrm{.} \nn
\end{flalign}

If no or more than one such tuples can be found, then declare error. It can be shown that for sufficiently large $n$, decoding is correct with high probability if
    \begin{flalign}
      R_{22} + \tilde{R}_{22} \leqslant & I(V;Z|UX_1T)\label{eq:OriAch5}\\
      R_{21} + \tilde{R}_{21} + R_{22} + \tilde{R}_{22} \leqslant & I(UV;Z|X_1T)\\
      \tilde{R_1} + R_{21} + \tilde{R}_{21} + R_{22} + \tilde{R}_{22} \leqslant & I(T UV;Z|X_1)\label{eq:OriAch6}\\
      R_1 + \tilde{R_1} + R_{21} + \tilde{R}_{21} + R_{22} + \tilde{R}_{22} \leqslant & I(TUVX_1;Z)\label{eq:OriAch7}
    \end{flalign}
\end{list}

Lemma \ref{NaiAchv} is thus proved by combining \eqref{eq:OriAch1}-\eqref{eq:OriAch7}.

\section{Proof of Theorem \ref{thr:OrigAchConv}}\label{ProofOrigConv}
Consider a $(2^{nR_1}, 2^{nR_2},n)$ code with an average error probability $P_e^{(n)}$. The probability distribution on $\mathcal{W}_1 \times \mathcal{W}_2 \times \mathcal{S}^n
\times \mathcal{X}_1^n \times \mathcal{X}_2^n \times \mathcal{Y}^n \times \mathcal{Z}^n$ is given by
\begin{equation} \label{JointlyDistConv}
\begin{split}
P&_{W_1W_2S^nX_1^nX_2^nY^nZ^n}=P_{W_1}P_{W_2}\left[{\prod_{i=1}^n{P_{S_i}}}\right]P_{X_1^n|W_1}P_{X_2^n|W_1W_2S^n}\prod_{i=1}^n{P_{Y_iZ_i|X_{1i}X_{2i}S_i}}.
\end{split}
\end{equation}
By Fano's inequality, we have
\begin{flalign}
H(W_1|Y^n) &\leqslant nR_1P_e^{(n)}+1 = n\delta_{1n} \nn \\
H(W_1W_2|Z^n) &\leqslant n(R_1+R_2)P_e^{(n)}+1 = n\delta_{2n}
\end{flalign}
where $\delta_{1n}$, $\delta_{2n} \to 0$ as $n\to +\infty$. Let $\delta_{n} = \delta_{1n} + \delta_{2n}$, which also satisfies that $\delta_{n} \to 0$ as $n\to +\infty$.

We define the following auxiliary random variables:
\begin{flalign}
T_i &= (W_1, S_{i+1}^n, X_1^n)\nn\\
U_i &=(T_i, Y^{i-1})\nn\\
V_i &= (T_i, W_2, Z^{i-1})
\end{flalign}
which satisfy the Markov chain conditions:
\begin{equation}
T_i \longleftrightarrow U_iV_i \longleftrightarrow X_{1i}X_{2i}S_i \longleftrightarrow Y_iZ_i
\end{equation}
for $i=1, \dotsi, n$.

We first bound $R_1$ and obtain
\begin{flalign}
nR_1 & \leqslant I(W_1;Y^n) + n\delta_{n}\nn\\
= & \sum_{i=1}^n [I(W_1 S_{i+1}^n ; Y^i)-I(W_1 S_i^n ; Y^{i-1})] + n\delta_{n}\nn\\
= & \sum_{i=1}^n [I(W_1 S_{i+1}^n ; Y^{i-1}) + I(W_1 S_{i+1}^n ; Y_i|Y^{i-1}) \nn \\
& \quad\; - I(W_1 S_{i+1}^n ; Y^{i-1}) - I(S_i; Y^{i-1}|W_1 S_{i+1}^n)] + n\delta_{n}\nn\\
= & \sum_{i=1}^n [I(W_1 S_{i+1}^n ; Y_i|Y^{i-1}) - I(S_i; Y^{i-1}|W_1 S_{i+1}^n)] + n\delta_{n}\nn\\
= & \sum_{i=1}^n [H(Y_i|Y^{i-1}) - H(Y_i|W_1 S_{i+1}^n Y^{i-1}) - H(S_i|W_1 S_{i+1}^n) + H(S_i|W_1 S_{i+1}^n Y^{i-1})] + n\delta_{n}\nn\\
\overset{(a)}{\leqslant} & \sum_{i=1}^n [H(Y_i) - H(Y_i|W_1 S_{i+1}^n Y^{i-1} X_1^n ) -(H(S_i|X_{1i}) + H(S_i|W_1 S_{i+1}^n Y^{i-1} X_1^n))] + n\delta_{n}\nn\\
\leqslant & \sum_{i=1}^n [I(T_i U_i X_{1i};Y_i) - I(T_i U_i;S_i|X_{1i})] + n\delta_{n}\label{eq:OrigConv1}
\end{flalign}
where $(a)$ follows because $X_1^n$ is a function of $W_1$.


We next bound $R_2$ and obtain
\begin{flalign}
nR_2 = & I(W_2;Z^n) + n\delta_n\leqslant I(W_2;Z^n|W_1) + n\delta_n\nn\\
= & \sum_{i=1}^n[I(W_2 S_{i+1}^n ; Z^i|W_1)-I(W_2 S_i^n;Z^{i-1}|W_1)] + n\delta_n\nn\\
= & \sum_{i=1}^n [I(W_2 S_{i+1}^n ; Z^{i-1}|W_1) + I(W_2 S_{i+1}^n ; Z_i|W_1 Z^{i-1})\nn\\
& \quad\; - I(W_2 S_{i+1}^n ; Z^{i-1}|W_1) - I(S_i; Z^{i-1}|W_1 W_2 S_{i+1}^n)] + n\delta_n\nn\\
= & \sum_{i=1}^n [I(W_2 S_{i+1}^n ; Z_i|W_1 Z^{i-1}) - I(S_i ; Z^{i-1}|W_1 W_2 S_{i+1}^n)] + n\delta_n\nn\\
= & \sum_{i=1}^n [H(Z_i|W_1 Z^{i-1}) - H(Z_i|W_1 W_2 S_{i+1}^n Z^{i-1})\nn\\
& \quad\; - H(S_i|W_1 W_2 S_{i+1}^n) + H(S_i|W_1 W_2 S_{i+1}^n Z^{i-1})] + n\delta_n\label{eq:conOriAch21}\\
= & \sum_{i=1}^n [H(Z_i|W_1 Z^{i-1} X_{1i}) - H(Z_i|W_1 W_2 S_{i+1}^n X_1^n Z^{i-1})\nn\\
& \quad\; - H(S_i|W_1 W_2 S_{i+1}^n X_{1i}) + H(S_i|W_1 W_2 S_{i+1}^n X_1^n Z^{i-1})] + n\delta_n\nn\\
\leqslant & \sum_{i=1}^n [H(Z_i|X_{1i}) - H(Z_i|X_{1i}T_i V_i)- H(S_i|X_{1i}) + H(S_i|X_{1i} T_i V_i)] + n\delta_n\nn\\
= & \sum_{i=1}^n [I(T_i V_i;Z_i|X_i) - I(T_i V_i;S_i|X_{1i})] + n\delta_n \textrm{.}\label{eq:OrigConv2}
\end{flalign}

We then bound the sum rate $R_1+R_2$ as follows.
\begin{flalign}
&n(R_1 + R_2)\nn\\
= & I(W_1W _2;Z^n) + n\delta_n\\
= & \sum_{i=1}^n [I(W_1 W_2 S_{i+1}^n ; Z^i)-I(W_1 W_2 S_i^n ; Z^{i-1})] + n\delta_n\nn\\
= & \sum_{i=1}^n [I(W_1 W_2 S_{i+1}^n ; Z^{i-1}) + I(W_1 W_2 S_{i+1}^n ; Z_i|Z^{i-1}) \nn\\
& \quad\; - I(W_1 W_2 S_{i+1}^n ; Z^{i-1}) - I(S_i; Z^{i-1}|W_1 W_2 S_{i+1}^n)] + n\delta_n\nn\\
= & \sum_{i=1}^n [I(W_1W_2S_{i+1}^n ;Z_i|Z^{i-1}) - I(S_i;Z^{i-1}|S_{i+1}^n W_1 W_2)] + n\delta_n\nn\\
= & \sum_{i=1}^n [H(Z_i|Z^{i-1}) - H(Z_i|W_1 W_2 S_{i+1}^n Z^{i-1}) \nn \\
& \quad\; - H(S_i|S_{i+1}^n W_1 W_2) + H(S_i|S_{i+1}^n W_1 W_2 Z^{i-1})] + n\delta_n\nn\\
\leqslant & \sum_{i=1}^n [H(Z_i) - H(Z_i|W_1 W_2 S_{i+1}^n X_1^n Z^{i-1}) - H(S_i|X_{1i}) + H(S_i|W_1 W_2 S_{i+1}^n X_1^n Z^{i-1})] + n\delta_n\nn\\
= & \sum_{i=1}^n [I(X_{1i} T_i V_i;Z_i) - I(T_i V_i;S_i|X_{1i})] + n\delta_n\label{OrigConv121}
\end{flalign}

\section{Proof of the Outer Bound for Theorem \ref{DegrConv}}\label{pr:DegrConv}

We define the following auxiliary random variables:
\begin{flalign}
T_i &= (W_1, S_{i+1}^n, X_1^n, Y^{i-1})\nn\\
V_i &= (T_i, W_2, Z^{i-1})
\end{flalign}
which satisfy the Markov chain conditions:
\begin{equation}
T_i \longleftrightarrow V_i \longleftrightarrow X_{1i}X_{2i}S_i \longleftrightarrow Y_i\longleftrightarrow Z_i
\end{equation}
for $i=1, \dotsi, n$.

By following the step similar to those in \eqref{eq:OrigConv1}, we obtain the following bound on $R_1$:
\begin{flalign}
nR_1 \leqslant  \sum_{i=1}^n [I(T_i X_{1i};Y_i) - I(T_i;S_i|X_{1i})] + n\delta_{n}\textrm{.}\label{eq:r1degraded}
\end{flalign}

We next derive a bound on $R_2$ by continuing to derive the bound \eqref{eq:conOriAch21} as follows:
\begin{flalign}
nR_2 \leqslant& \sum_{i=1}^n [H(Z_i|W_1 Z^{i-1}) - H(Z_i|W_1 W_2 S_{i+1}^n Z^{i-1})\nn\\
& - H(S_i|W_1 W_2 S_{i+1}^n) + H(S_i|W_1 W_2 S_{i+1}^n Z^{i-1})] + n\delta_n\nn\\
= & \sum_{i=1}^n [H(Z_i|W_1 Z^{i-1} X_{1i}) - H(Z_i|W_1 W_2 S_{i+1}^n X_1^n Y^{i-1} Z^{i-1})\nn\\
& - H(S_i|W_1 W_2 S_{i+1}^n X_{1i}) + H(S_i|W_1 W_2 S_{i+1}^n X_1^n Y^{i-1} Z^{i-1})] + n\delta_n\nn\\
\leqslant & \sum_{i=1}^n [H(Z_i|X_{1i}) - H(Z_i|X_{1i}T_i V_i)- H(S_i|X_{1i}) + H(S_i|X_{1i} T_i V_i)] + n\delta_n\nn\\
= & \sum_{i=1}^n [I(T_i V_i;Z_i|X_i) - I(T_i V_i;S_i|X_{1i})] + n\delta_n \textrm{.}
\end{flalign}

\section{Proof of the Converse for Theorem \ref{DetmDegrConv}}\label{ProofDetmDegrConv}

We define the auxiliary random variable $T_i=(W_1 S_{i+1}^n X_1^n Y^{i-1} )$, which satisfies the Markov chain:
\begin{equation}
T_i \leftrightarrow X_{1i}X_{2i}S_i \leftrightarrow Y_iZ_i, \quad \text{for  } i=1, \dotsi, n.
\end{equation}

Following \eqref{eq:r1degraded}, we obtain
\begin{flalign}
nR_1 \leqslant  \sum_{i=1}^n [I(T_i X_{1i};Y_i) - I(T_i;S_i|X_{1i})] + n\delta_{n}\textrm{.}\nn
\end{flalign}

%

We next bound $R_2$ as follows.
\begin{flalign}
nR_2 = &I(W_2;Z^n) + n\delta_n \nn\\
\leqslant &I(W_2;Z^n|W_1 S^n X_1^n) + n\delta_n\nn\\
= & \sum_{i=1}^n [I(W_2;Z_i|W_1 S^n X_1^n Z^{i-1})] + n\delta_n\nn\\
\leqslant & \sum_{i=1}^n H(Z_i|W_1 S^n X_1^n Z^{i-1})+ n\delta_n\nn\\
\overset{(a)}{=} & \sum_{i=1}^n H(Z_i|W_1 S^n  X_{1}^n Y^{i-1} Z^{i-1})+ n\delta_n\nn\\
\leqslant & \sum_{i=1}^n H(Z_i|W_1 S_{i+1}^n  X_{1}^n Y^{i-1} S_i)+ n\delta_n\nn\\
\leqslant & \sum_{i=1}^n H(Z_i|X_{1i} T_i S_i)+ n\delta_n\label{eq:convDegrR21}
\end{flalign}
where $(a)$ follows due to the degradedness condition \eqref{eq:cond5}.

We then derive another bound on $R_2$ by continuing to derive the bound \eqref{eq:conOriAch21} as follows:
\begin{equation}
\begin{split} \nonumber
nR_2
\leqslant & \sum_{i=1}^n [H(Z_i|W_1 Z^{i-1}) - H(Z_i|W_1 W_2 S_{i+1}^n Z^{i-1})\\
& - H(S_i|W_1 W_2 S_{i+1}^n) + H(S_i|W_1 W_2 S_{i+1}^n Z^{i-1})] + n\delta_n\\
= & \sum_{i=1}^n [H(Z_i|W_1 X_1^n Z^{i-1}) - H(S_i|W_1 W_2 X_1^n S_{i+1}^n)+ H(S_i|W_1 W_2 X_1^n S_{i+1}^n Y^{i-1} Z^{i-1} Z_i)\\
 & + I(Z_i;S_i|W_1 W_2 S_{i+1}^n Z^{i-1}) - H(Z_i|W_1 W_2 S_{i+1}^n Z^{i-1})] + n\delta_n\\
\leqslant & \sum_{i=1}^n [H(Z_i|X_{1i}) - H(S_i|X_{1i}) + H(S_i|X_{1i} T_i Z_i)] + n\delta_n\\
= & \sum_{i=1}^n [H(Z_i|X_{1i}) - I(T_i Z_i;S_i|X_{1i})] + n\delta_n \textrm{.}
\end{split}
\end{equation}
%


\section{Proof of the Outer Bound \eqref{eq:outerbothstate}}\label{pr:ConWithSn}

Consider a $(2^{nR_1}, 2^{nR_2},n)$ code with an average error probability $P_e^{(n)}$. The probability distribution on $\mathcal{W}_1 \times \mathcal{W}_2 \times \mathcal{S}^n
\times \mathcal{X}_1^n \times \mathcal{X}_2^n \times \mathcal{Y}^n \times \mathcal{Z}^n$ is given by
\begin{equation}
P_{W_1W_2S^nX_1^nX_2^nY^nZ^n}=P_{W_1}P_{W_2}\left[{\prod_{i=1}^n{P_{S_i}}}\right]P_{X_1^n|W_1}P_{X_2^n|W_1W_2S^n}\prod_{i=1}^n{P_{Y_iZ_i|X_{1i}X_{2i}S_i}}.
\end{equation}
By Fano's inequality, we have
\begin{flalign}
H(W_1|Y^n) & \leqslant nR_1P_e^{(n)}+1 = n\delta_{1n} \nn\\
H(W_1W_2|S^n Z^n) & \leqslant n(R_1+R_2)P_e^{(n)}+1 = n\delta_{2n}
\end{flalign}
where $\delta_{1n}$, $\delta_{2n} \to 0$ as $n\to +\infty$. Let $\delta_{n} = \delta_{1n} + \delta_{2n}$, which also satisfies that $\delta_{n} \to 0$ as $n\to +\infty$.

We define the following auxiliary random variables:
\begin{flalign}
K_i &=(W_1, S_{i+1}^n, X_1^n, Y^{i-1} )\nn\\
T_i &= Z_{i+1}^n
\end{flalign}
which satisfies the Markov chain condition:
\begin{equation}
K_iT_i \leftrightarrow X_{1i}X_{2i}S_i \leftrightarrow Y_iZ_i
\end{equation}
for $i=1, \dotsi, n$.

The following bound on $R_1$ follows the same steps as in \eqref{eq:OrigConv1} in Appendix \ref{ProofOrigConv}, and we have
\begin{equation}
nR_1 \leqslant \sum_{i=1}^n [I(K_iX_{1i};Y_i) - I(K_i;S_i|X_{1i})] + n\delta_{n}.\label{eq:dmr1}
\end{equation}

We next bound $R_2$ and obtain
\begin{flalign}
nR_2 & \leqslant I(W_2;Z^n S^n) + n\delta_{n} \leqslant I(W_2;Z^nS^n W_1) + n\delta_{n} \nn\\
& \leqslant I(W_2;Z^n| W_1 S^n) + n\delta_{n}\nn\\
&= \sum_{i=1}^n I(W_2 ; Z_i|Z_{i+1}^n S^n W_1 X_1^n) + n\delta_{n}\nn\\
&\leqslant \sum_{i=1}^n [H(Z_i|S_i X_{1i}) - H(Z_i|W_2 Z_{i+1}^n S^n W_1 X_1^n X_{2i})] + n\delta_{n}\nn\\
&\leqslant \sum_{i=1}^n [H(Z_i|S_i X_{1i}) - H(Z_i|S_i X_{1i}X_{2i})] + n\delta_{n}\nn\\
&= \sum_{i=1}^n I(X_{2i};Z_i|S_i X_{1i})+ n\delta_{n}. \label{Eq:OuterWiSnR2}
\end{flalign}

We further bound $R_1 + R_2$ as follows:
\begin{flalign}
n(R_1 + R_2) & \leqslant I(W_1 W_2;Z^n S^n) + n\delta_{n}\nn\\
&= \sum_{i=1}^n I(W_2 W_1; Z_i|Z_{i+1}^n S^n ) + n\delta_{n}\nn\\
&\leqslant \sum_{i=1}^n [H(Z_i|S_i ) - H(Z_i|W_2 Z_{i+1}^n S^n W_1 X_{1i}X_{2i})] + n\delta_{n}\nn\\
&\leqslant \sum_{i=1}^n [H(Z_i|S_i ) - H(Z_i|S_i X_{1i}X_{2i})] + n\delta_{n}\nn\\
&=  \sum_{i=1}^n I(X_{1i} X_{2i};Z_i|S_i)+ n\delta_{n} \textrm{.}\label{Eq:OuterWiSnR12}
\end{flalign}

We proceed to derive an alternative bound on $R_1+R_2$ as follows:
\begin{flalign}
n(R_1 + R_2) & \leqslant I(W_1; Y^n) + I(W_2;Z^n S^n) + n\delta_{n} \nn \\
& \leqslant I(W_1; Y^n) + I(W_2;Z^n S^n|W_1)+ n\delta_{n} \label{eq:bothstateR12}
\end{flalign}

The first term in \eqref{eq:bothstateR12} can be bounded as follows:
\begin{flalign}
& I(W_1;Y^n) \nn \\
& = \sum_{i=1}^n I(W_1;Y_i|Y^{i-1}) \nn\\
& \leqslant \sum_{i=1}^n I(W_1 Y^{i-1}; Y_i) \nn\\
& = \sum_{i=1}^n \left[I(W_1 Y^{i-1}S_{i+1}^n Z_{i+1}^n; Y_i) - I(S_{i+1}^n Z_{i+1}^n ; Y_i| W_1 Y^{i-1})\right] \nn\\
& = \sum_{i=1}^n \left[I(W_1 Y^{i-1}S_{i+1}^n Z_{i+1}^n; Y_i) - I(S_{i} Z_{i} ; Y^{i-1}| W_1 S_{i+1}^n Z_{i+1}^n)\right] \nn\\
& = \sum_{i=1}^n \left[I(W_1 Y^{i-1}S_{i+1}^n Z_{i+1}^n; Y_i) - I(S_{i} Z_{i} ; Y^{i-1} W_1 S_{i+1}^n Z_{i+1}^n)+ I(W_1 S_{i+1}^n Z_{i+1}^n;S_i Z_i)\right]  \nn\\
& = \sum_{i=1}^n [I(W_1 Y^{i-1}S_{i+1}^n Z_{i+1}^n; Y_i) - I(S_{i} ; Y^{i-1} W_1 S_{i+1}^n Z_{i+1}^n)\nn\\
& \hspace{1.2cm} + I(W_1 S_{i+1}^n Z_{i+1}^n;S_i Z_i)- I(Z_{i} ; Y^{i-1} W_1 S_{i+1}^n Z_{i+1}^n|S_i)] \nn\\
&= \sum_{i=1}^n [I(T_i K_i X_{1i}; Y_i) - I(T_i K_i X_{1i}; S_{i} )\nn\\
& \hspace{1.2cm} + I(W_1 S_{i+1}^n Z_{i+1}^n;S_i Z_i)- I(Z_{i} ; Y^{i-1} W_1 S_{i+1}^n Z_{i+1}^n|S_i)] \label{eq:R1}
\end{flalign}

We next consider the last two terms in \eqref{eq:R1} together with the second bound in \eqref{eq:bothstateR12} as follows:
\begin{flalign}
& I(W_2;Z^n S^n|W_1)+ \sum_{i=1}^n  \left[I(W_1 S_{i+1}^n Z_{i+1}^n;S_i Z_i)- I(Z_{i} ; Y^{i-1} W_1 S_{i+1}^n Z_{i+1}^n|S_i) \right] \nn\\
& =  \sum_{i=1}^n \left[I(W_2;Z_i S_i|W_1 S_{i+1}^n Z_{i+1}^n)+ I(W_1 S_{i+1}^n Z_{i+1}^n;S_i Z_i)- I(Z_{i} ; Y^{i-1} W_1 S_{i+1}^n Z_{i+1}^n|S_i)\right]  \nn\\
& \overset{(a)}{=} \sum_{i=1}^n [I(W_1 W_2 S_{i+1}^n Z_{i+1}^n;S_i Z_i)+ I(S^{i-1};S_i Z_i|W_1 W_2 S_{i+1}^n Z_{i+1}^n)- I(S_{i+1}^n Z_{i+1}^n;S_i|W_1 W_2 S^{i-1}) \nn\\
& \hspace{1.2cm} - I(Z_{i} ; Y^{i-1} W_1 S_{i+1}^n Z_{i+1}^n|S_i) ] \nn\\
& =  \sum_{i=1}^n [I(W_1 W_2 S_{i+1}^n S^{i-1} Z_{i+1}^n;S_i Z_i) - I(S_{i+1}^n Z_{i+1}^n;S_i|W_1 W_2 S^{i-1}) \nn \\
& \hspace{1.2cm} - I(Z_{i} ; Y^{i-1} W_1 S_{i+1}^n Z_{i+1}^n|S_i)] \nn \\
& =  \sum_{i=1}^n [I(W_1 W_2 S_{i+1}^n S^{i-1} Z_{i+1}^n;S_i Z_i) - I(S_{i+1}^n Z_{i+1}^n W_1 W_2 S^{i-1};S_i) - I(Z_{i} ; Y^{i-1} W_1 S_{i+1}^n Z_{i+1}^n|S_i) ] \nn\\
& =  \sum_{i=1}^n \left[I(W_1 W_2 S_{i+1}^n S^{i-1} Z_{i+1}^n; Z_i|S_i)  - I(Z_{i} ; Y^{i-1} W_1 S_{i+1}^n Z_{i+1}^n|S_i) \right]\nn\\
& \leqslant \sum_{i=1}^n \left[I(W_1 W_2 S_{i+1}^n S^{i-1} Z_{i+1}^n; Z_i|S_i)  - I(Z_{i} ; Y^{i-1} W_1 S_{i+1}^n Z_{i+1}^n|S_i)\right]  \nn\\
& \leqslant \sum_{i=1}^n  \left[H(Z_i|S_i Y^{i-1} W_1 X_1^nS_{i+1}^n Z_{i+1}^n)- H(Z_i|S_i Y^{i-1} W_1X_1^n W_2 S_{i+1}^n S^{i-1} Z_{i+1}^n X_{2i})\right]\nn\\
& \leqslant \sum_{i=1}^n  \left[H(Z_i|S_i Y^{i-1} W_1 X_1^nS_{i+1}^n Z_{i+1}^n)- H(Z_i|S_i Y^{i-1} W_1X_1^nS_{i+1}^n Z_{i+1}^n X_{2i})\right]\nn\\
& =  \sum_{i=1}^n I(X_{2i};Z_i|X_{1i} T_i  K_i S_i)\label{eq:conwiSnR12}
\end{flalign}
where (a) follows from Csiszar-Korner's Sum Identity \cite{Csiszar81}.

Therefore, substituting \eqref{eq:R1} and \eqref{eq:conwiSnR12} into \eqref{eq:bothstateR12}, we obtain
\begin{equation}
  n(R_1+R_2) \leqslant \sum_{i=1}^n [I( T_i K_i X_{1i};Y_i) - I(T_i K_i;S_i|X_{1i}) + I (X_{2i};Z_i|X_{1i} T_i K_i S_i)]+ n\delta_{n}.
\end{equation}

\section{Proof for Theorem \ref{thr:GauWithSn1Con}}\label{apx:ProofGauWithSn1Con}

For the Gaussian channel, if $|a|\leqslant 1$, it satisfies the condition \eqref{eq:cond1}. For these channels, we first prove the following bounds.
\begin{flalign}
nR_1 & \leqslant \sum_{i=1}^n [I(U_iX_{1i};Y_i) - I(U_i;S_i|X_{1i})] +n\delta_n \label{eq:trdm1r1} \\
nR_2 & \leqslant \sum_{i=1}^nI(X_{2i};Z_i|U_i X_{1i} S_i)+ n\delta_n \label{eq:trdm1r2} \\
n(R_1 + R_2) & \leqslant \sum_{i=1}^n I(X_{1i} X_{2i};Z_i|S_i)+ n\delta_n \label{eq:trdm1r12}
\end{flalign}
where $U_i=(W_1 S_{i+1}^n X_1^n Y^{i-1} )$ for $i=1,\ldots,n$.


The bounds \eqref{eq:trdm1r1} and \eqref{eq:trdm1r12} follow from \eqref{eq:dmr1} and \eqref{Eq:OuterWiSnR12}, respectively. We then bound $R_2$ as follows:
\begin{flalign}
nR_2 = & I(W_2;Z^n S^n) + n \delta_n\nn\\
  \leqslant & I(W_2;Z^n S^n|W_1)+ n\delta_n\nn\\
  = & I(W_2;Z^n|W_1 S^n)+ n\delta_n\nn\\
  = & \sum_{i=1}^n{I(W_2;Z_i|W_1 S^n Z^{i-1})}+ n\delta_n\nn\\
  = & \sum_{i=1}^n{[H(Z_i|W_1 S^n Z^{i-1}) - H(Z_i|W_1 W_2 S^n Z^{i-1})]}+ n\delta_n\nn\\
  \overset{(a)}{=} &\sum_{i=1}^n{[H(Z_i|W_1 S^n X_1^n Y^{i-1} Z^{i-1})} - H(Z_i|W_1 W_2 S^n Z^{i-1} X_1^n Y^{i-1} )]+ n\delta_n\nn\\
  \overset{(b)}{\leqslant} &\sum_{i=1}^n[H(Z_i|W_1 S_{i+1}^n X_1^n Y^{i-1} S_i) - H(Z_i|W_1 S_{i+1}^n S_i X_1^n Y^{i-1} X_{2i})]+ n\delta_n\nn\\
  \leqslant & \sum_{i=1}^n{[H(Z_i|S_i X_{1i} U_i) - H(Z_i|S_i X_{1i} U_i X_{2i})]}+ n\delta_n\nn\\
  \leqslant & \sum_{i=1}^nI(X_{2i};Z_i|U_i X_{1i} S_i)+ n\delta_n
\end{flalign}
where $(a)$ follows from the condition \eqref{eq:cond1} and the fact that $X_1^n$ is a function of $W_1$, and (b) follows from the fact that given $X_{1i}$, $X_{2i}$, and $S_i$, $Z_i$ is independent of all other variables.

%

We now further derive the bounds \eqref{eq:trdm1r1}-\eqref{eq:trdm1r12} for Gaussian channels. We first consider the bound on $R_1$ as follows:
  \begin{flalign}
R_1 & \leqslant \frac{1}{n} \sum_{i=1}^n [I(X_{1i} U_i;Y_i) - I(U_i;S_i|X_{1i})]\ \nn\\
& = \frac{1}{n} \sum_{i=1}^n [h(Y_i)- h(Y_i|X_{1i}U_i) -h(S_i|X_{1i}) + h(S_i|X_{1i}U_i) ]\nn\\
& = \frac{1}{n} \sum_{i=1}^n [h(Y_i)- h(Y_i|X_{1i}U_i S_i) -I(S_i;Y_i|X_{1i}U_i) -h(S_i|X_{1i}) + H(S_i|X_{1i}U_i)]\nn\\
& = \frac{1}{n} \sum_{i=1}^n [h(Y_i)- h(Y_i|X_{1i}U_i S_i) -h(S_i|X_{1i}) + h(S_i|X_{1i}U_i Y_i)] \nn\\
& \leqslant \frac{1}{n} \sum_{i=1}^n [h(Y_i)- h(Y_i|X_{1i}U_i S_i) -h(S_i) + h(S_i|X_{1i} Y_i)]\label{eq:GauwithSn1Con1}
  \end{flalign}

The first term in \eqref{eq:GauwithSn1Con1} can be derived as:
\begin{flalign}
&\frac{1}{n} \sum_{i=1}^n h(Y_i) \nn\\
& \overset{(a)}{\leqslant} \frac{1}{2n} \sum_{i=1}^n \log 2\pi e(E(X_{1i}+aX_{2i}+S_i+N_i)^2)\nn\\
& \leqslant \frac{1}{2n} \sum_{i=1}^n \log 2\pi e\Big(E[X_{1i}^2]+2aE(X_{1i}X_{2i})+a^2E[X_{2i}^2]+ 2aE(X_{2i}S_i) + E[S_i^2] + E[N_i^2])\Big)\nn\\
& \overset{(b)}{\leqslant} \frac{1}{2} \log 2\pi e \Bigg(\frac{1}{n}\sum_{i=1}^n E[X_{1i}^2]+\frac{2a}{n}\sum_{i=1}^n E(X_{1i}X_{2i})+\frac{a^2}{n}\sum_{i=1}^n E[X_{2i}^2]+ \frac{2a}{n}\sum_{i=1}^n E(X_{2i}S_i) \nn \\
& \hspace{2.5cm}+ \frac{1}{n}\sum_{i=1}^n E[S_i^2] + \frac{1}{n}\sum_{i=1}^n E[N_i^2])\Bigg)\nn\\
& \leqslant \frac{1}{2} \log 2\pi e\left(P_1 + a^2P_2 + Q +1 + \frac{2a}{n} \sum_{i=1}^n E(X_{1i}X_{2i}) +  \frac{2a}{n} \sum_{i=1}^n E(X_{2i}S_i)\right)\nn\\
& \leqslant \frac{1}{2} \log 2\pi e\left(P_1 + a^2P_2 + Q +1 + 2a\rho_{21}\sqrt{P_1P_2} + 2a\rho_{2s}\sqrt{P_2Q}\right) \label{eq:gauss1r1-1}
\end{flalign}
where $\rho_{21} = \frac{\frac{1}{n} \sum_{i=1}^n E(X_{1i}X_{2i})}{\sqrt{P_1P_2}}$ and $\rho_{2s} = \frac{\frac{1}{n} \sum_{i=1}^n E(X_{2i}S_i)}{\sqrt{P_2Q}}$. In the above derivation, $(a)$ follows from the fact that the Gaussian distribution maximizes the entropy given the variance of the random variable, and $(b)$ follows from the concavity of the logarithm function and Jensen's inequality.


The second term in \eqref{eq:GauwithSn1Con1} can be bounded as:
\begin{equation}
   \frac{1}{n} \sum_{i=1}^n h(Y_i|X_{1i} X_{2i} S_i)   \leqslant \frac{1}{n} \sum_{i=1}^n h(Y_i|X_{1i} U_i S_i)\leqslant \frac{1}{n} \sum_{i=1}^n h(Y_i|X_{1i} S_i)
\end{equation}

For the left-hand side, we have
\begin{equation}\label{eq:gauss1r1-2}
  \frac{1}{n} \sum_{i=1}^n h(Y_i|X_{1i} X_{2i} S_i) =  \frac{1}{2} \log 2\pi e\textrm{.}
\end{equation}

For the right-hand side, by setting $\alpha = a \rho_{21} \sqrt{\frac{P_2}{P_1}}$ and $\beta = a \rho_{2S} \sqrt{\frac{P_2}{Q}}$, we have
\begin{equation}\nonumber
  \begin{split}
&\frac{1}{n} \sum_{i=1}^n h(Y_i|X_{1i} S_i) \nn\\
& = \frac{1}{n} \sum_{i=1}^n h(X_{1i} + aX_{2i} +S_i +N_{1i}|S_i X_{1i})\\
&= \frac{1}{n} \sum_{i=1}^n h(aX_{2i} +N_{1i}- \alpha X_{1i}- \beta S_i|S_i X_{1i})\\
&\leqslant \frac{1}{n} \sum_{i=1}^n h(aX_{2i} +N_{1i}- \alpha X_{1i}- \beta S_i)\\
&\leqslant \frac{1}{2n} \sum_{i=1}^n \log(2 \pi e E[(aX_{2i} +N_{1i}- \alpha X_{1i}- \beta S_i)^2])\\
&\leqslant \frac{1}{2} \log 2\pi e\left(a^2 P_2 +1 +\alpha^2 P_1 + \beta^2 Q - 2a \alpha \frac{1}{n} \sum_{i=1}^n E[X_{1i}X_{2i}] - 2a \beta \frac{1}{n} \sum_{i=1}^n E[X_{2i}S_i]\right)\\
&= \frac{1}{2} \log 2\pi e\left(a^2 P_2 +1 - a^2 \rho_{2S}^2 P_2 - a^2 \rho_{21}^2 P_2 \right) .
  \end{split}
\end{equation}

Therefore, there exists $0\leqslant P_2'' \leqslant (1-\rho_{2S}^2  - \rho_{21}^2 ) P_2$ such that
 \begin{equation}\label{eq:PreForEPI}
   \frac{1}{n} \sum_{i=1}^n h(Y_i|X_{1i} U_i S_i) = \frac{1}{2} \log 2\pi e(1+ a^2 P_2'') \;.
 \end{equation}

The third term in \eqref{eq:GauwithSn1Con1} is given by
\begin{equation}\label{eq:gauss1r1-3}
  \frac{1}{n} \sum_{i=1}^n h(S_i)= \frac{1}{2} \log 2\pi e Q \;.
\end{equation}

Finally, for the fourth term in \eqref{eq:GauwithSn1Con1}, we first define $\alpha' = \frac{-a\rho_{21}\sqrt{P_2P_1}(a\rho_{2s}\sqrt{P_2Q}+Q)}{\left(a^2 (1-\rho_{21}^2)P_2 + Q + 2a \rho_{2s}\sqrt{P_2Q}+1\right)P_1}$ and $\beta' = -\frac{P_1}{a\rho_{21}\sqrt{P_1P_2}} \alpha'$, and then have
\begin{flalign}
\frac{1}{n} \sum_{i=1}^n & h(S_i|X_{1i} Y_i) \nn\\
 = &\frac{1}{n} \sum_{i=1}^n h(S_i|X_{1i}\textrm{, } X_{1i} + aX_{2i} +S_i +N_{1i}) \nn \\
=& \frac{1}{n} \sum_{i=1}^n h(S_i - \alpha' X_1 - \beta' (aX_{2i} + S_i + N_{1i})|X_{1i}\textrm{, } X_{1i} + aX_{2i} +S_i +N_{1i}) \nn\\
\leqslant & \frac{1}{n} \sum_{i=1}^n h(S_i - \alpha' X_{1i} - \beta' (aX_{2i} + S_i + N_{1i})) \nn \\
= & \frac{1}{n} \sum_{i=1}^n \log\Big(2 \pi e E(S_i - \alpha' X_{1i} - \beta' (aX_{2i} + S_i + N_{1i}))^2\Big) \nn \\
\leqslant & \frac{1}{2} \log 2\pi e\Bigg(Q+ \alpha^{'2} P_1 + a^2\beta^{'2} P_2 + \beta^{'2} Q + 2a\beta^{'2}\frac{1}{n} \sum_{i=1}^n E(X_{2i}S_i) +\beta^{'2} \nn \\
&+ 2\alpha' \beta' a \frac{1}{n} \sum_{i=1}^n E(X_{1i}X_{2i})-2\beta' a \frac{1}{n} \sum_{i=1}^n E(X_{2i}S_i)-2\beta' Q \Bigg) \nn \\
\leqslant & \frac{1}{2} \log 2\pi e \frac{(a^2(1-\rho_{21}^2 - \rho_{2s}^2)P_2 +1)Q}{a^2(1-\rho_{21}^2)P_2 +2a\rho_{2s}\sqrt{P_2Q}+Q+1} \label{eq:gauss1r1-4}
  \end{flalign}

Substituting the above four terms into \eqref{eq:GauwithSn1Con1}, we obtain
\[ R_1 \leqslant \frac{1}{2}\log\left(1+\frac{P_1+2a\rho_{21}\sqrt{P_1P_2}+a^2\rho_{21}^2P_2}{a^2(1-\rho_{21}^2)P_2+2a\rho_{2s}\sqrt{P_2Q}+Q+1}\right)+\frac{1}{2}\log\left(1+\frac{a^2 P_2'}{a^2 P_2''+1}\right)\]
where $P_2'=(1-\rho_{21}^2-\rho_{2s}^2)P_2-P_2''$.

We then bound $R_2$ by further deriving \eqref{eq:trdm1r2}. When $a\leqslant 1$, we have $Y_i=aZ_i+(1-ab)X_{1i}+(1-ac)S_i+N_i'$, where $N_i' \sim \mathcal{N}(0,1-a^2)$. By applying the conditional entropy power inequality\cite{Blachman65}, we have
\begin{flalign}
2^{2h(Y_i|U_i S_i X_{1i})} =& 2^{2h(aZ_i+(1-ab)X_{1i}+(1-ac)S_i+N_i'|U_i S_i X_{1i})}\nn\\
=& 2^{2h(aZ_i + N_i'|U_i S_i X_{1i})}\nn\\
\geqslant &  2^{2h(aZ_i |U_i S_i X_{1i})} + 2^{2h(N_i'|U_i S_i X_{1i})}\nn\\
=& 2^{2h(Z_i |U_i S_i X_{1i})+\log(a^2)}+2\pi e (1-a^2)\textrm{.}
\end{flalign}

Thus,
\begin{flalign}
  \frac{1}{n} \sum_{i=1}^n h(Z_i |U_i S_i X_{1i}) \leqslant & \frac{1}{n} \sum_{i=1}^n \frac{1}{2} \log\left(\frac{2^{2h(Y_i|U_i S_i X_{1i})}- 2 \pi e (1-a^2)}{a^2}\right)\nn\\
 \overset{(a)}{\leqslant} & \frac{1}{2} \log\left(\frac{2^{2 \frac{1}{n} \sum_{i=1}^n h(Y_i|U_i S_i X_{1i})}- 2 \pi e (1-a^2)}{a^2}\right)\nn\\
  \overset{(b)}{=} & \frac{1}{2} \log(2 \pi e (1+P_2''))
\end{flalign}
where $(a)$ follows because $\log \left(2^x-b\right)$ is concave for $b \geqslant 0$, and $(b)$ follows from \eqref{eq:PreForEPI}.

Therefore, we have
\begin{equation}
  \begin{split}
R_2 & \leqslant \frac{1}{n} \sum_{i=1}^n I(X_{2i};Z_i| X_{1i}S_i U_i)\\
& = \frac{1}{n} \sum_{i=1}^n \left[h(Z_i| X_{1i}S_i U_i)-h(Z_i| X_{1i}S_i X_{2i})\right]\\
&\leqslant \frac{1}{2} \log(2 \pi e (1+P_2'')) - \frac{1}{2} \log(2 \pi e )\\
&= \frac{1}{2} \log(1+P_2'')\; .
  \end{split}
\end{equation}

We finally bound $R_1 + R_2$ by further deriving \eqref{eq:trdm1r12}. We set $\alpha'' = \rho_{2s} \sqrt{\frac{P_2}{Q}}$, and have
\begin{flalign}
R_1 + R_2 \leqslant& \frac{1}{n} \sum_{i=1}^n I(X_{1i} X_{2i};Z_i|S_i )\nn\\
=& \frac{1}{n} \sum_{i=1}^n [h(Z_i| S_i)-h(Z_i| X_{1i}S_i X_{2i})]\nn\\
=& \frac{1}{n} \sum_{i=1}^n h(bX_{1i} + X_{2i} +cS_i +N_{1i}|S_i)- \frac{1}{2} \log 2\pi e\nn\\
=& \frac{1}{n} \sum_{i=1}^n h(bX_{1i} + X_{2i} +N_{1i}- \alpha''  S_i|S_i)- \frac{1}{2} \log 2\pi e\nn\\
\leqslant& \frac{1}{n} \sum_{i=1}^n h(bX_{1i} + X_{2i} +N_{1i}- \alpha''  S_i)- \frac{1}{2} \log 2\pi e\nn\\
\leqslant& \frac{1}{n} \sum_{i=1}^n \log(2 \pi e E(bX_{1i} + X_{2i} +N_{1i}- \alpha''  S_i)^2)- \frac{1}{2} \log 2\pi e\nn\\
\leqslant& \frac{1}{2} \log 2\pi e\left(b^2P_1 + P_2 +1 + \alpha^{''2} Q + 2 b \frac{1}{n} \sum_{i=1}^n E[X_{1i}X_{2i}] - 2\alpha'' \frac{1}{n} \sum_{i=1}^n E[X_{2i}S_i]\right)\nn\\
&- \frac{1}{2} \log 2\pi e\nn\\
=& \frac{1}{2} \log (b^2 P_1 + P_2 +1 + 2b \rho_{21}\sqrt{P_1P_2}-  \rho_{2s}^2 P_2)\textrm{.} \label{eq:gauss1r12}
\end{flalign}

\section{Proof of Lemma \ref{lm:DegrConvWithSn2}}\label{apx:ProofDMCGauWithSn2}


Following \eqref{Eq:OuterWiSnR2} and \eqref{Eq:OuterWiSnR12}, we obtain
\begin{flalign}
nR_2 \leqslant & \sum_{i=1}^n I(X_{2i};Z_i|S_i X_i)+ n\delta_{n} \label{eq:trdm2r2}\\
n(R_1 + R_2) \leqslant & \sum_{i=1}^n I(X_{1i} X_{2i};Z_i|S_i)+ n\delta_{n}\textrm{.} \label{eq:trdm2r121}
\end{flalign}

We then prove an alternative bound on $R_1+R_2$ as follows:
\begin{flalign}
n(& R_1+ R_2) \nn\\
\leqslant & I(W_1; Y^n) + I(W_2;Z^n|S^n)+ n\delta_{n}\nn\\
\leqslant & I(W_1; Y^n) + I(W_2;Z^n|S^n W_1)+ n\delta_{n}\nn\\
= & I(W_1; Y^n) + H(W_2|S^n W_1)- H(W_2|S^n W_1 Z^n)+ n\delta_{n}\nn\\
\leqslant & I(W_1; Y^n) + H(W_2|S^n W_1)- H(W_2|S^n W_1 Z^n Y^n X_1^n)+ n\delta_{n}\nn\\
\overset{(a)}{=} & I(W_1; Y^n) + H(W_2|S^n W_1)- H(W_2|S^n W_1 Y^n)+ n\delta_{n}\nn\\
= & I(W_1; Y^n) + I(W_2;Y^n|S^n W_1)+ n\delta_{n}\nn\\
= & \sum_{i=1}^n [H(Y_i|Y^{i-1}) - H(Y_i|W_1 Y^{i-1}) + H(Y_i|S^n W_1 Y^{i-1}) - H(Y_i|S^n W_1 W_2 Y^{i-1})]+ n\delta_{n}\nn\\
= & \sum_{i=1}^n [H(Y_i|Y^{i-1}) - H(Y_i|X_{1i}) + H(Y_i|X_{1i}) - H(Y_i|W_1 Y^{i-1})\nn\\
&\quad \quad+ H(Y_i|S^n W_1 Y^{i-1}) - H(Y_i|S^n W_1 W_2 Y^{i-1})]+ n\delta_{n}\nn\\
\leqslant & \sum_{i=1}^n [H(Y_i) - H(Y_i|X_{1i}) - H(Y_i|S^n X_{1i}X_{2i}W_1 W_2 Y^{i-1})\nn\\
&\quad\quad + H(Y_i|X_{1i}) - I(Y_i;S^n|W_1 Y^{i-1})]+ n\delta_{n}\nn\\
= & \sum_{i=1}^n [I(X_{1i};Y_i) - H(Y_i|S_iX_{1i}X_{2i}) + H(Y_i|X_{1i})] - I(Y^n;S^n|W_1)+ n\delta_{n}\nn\\
= & \sum_{i=1}^n [I(X_{1i};Y_i)  - H(Y_i|S_iX_{1i}X_{2i}) + H(Y_i|X_{1i})] - H(S^n) + H(S^n|Y^n W_1)+ n\delta_{n}\nn\\
= & \sum_{i=1}^n [I(X_{1i};Y_i)  - H(Y_i|S_iX_{1i}X_{2i}) + H(Y_i|X_{1i}) - H(S_i) + H(S_i|Y^n W_1 S_{i+1}^n)]+ n\delta_{n}\nn\\
\leqslant & \sum_{i=1}^n [I(X_{1i};Y_i)  - H(Y_i|S_iX_{1i}X_{2i}) + H(Y_i|X_{1i}) - H(S_i) + H(S_i|Y_i X_{1i})]+ n\delta_{n}\nn\\
= & \sum_{i=1}^n [I(X_{1i};Y_i)  - H(Y_i|S_iX_{1i}X_{2i}) + H(Y_i|X_{1i}) - I(S_i;Y_i|X_{1i})]+ n\delta_{n}\nn\\
= & \sum_{i=1}^n [I(X_{1i};Y_i)  - H(Y_i|S_iX_{1i}X_{2i}) + H(Y_i|S_i X_{1i})] + n\delta_{n}\nn\\
= & \sum_{i=1}^n [I(X_{1i};Y_i)  + I(X_{2i};Y_i|S_i X_{1i})]+ n\delta_{n} \label{eq:trdm2r122}
\end{flalign}
where (a) follows due to the condition \eqref{eq:cond2}.

\section{Proof of the Converse for Theorem \ref{thr:GauWithSn2Con}}\label{apx:ProofGauWithSn2Con}

Based on the outer bound derived in Appendix \ref{apx:ProofDMCGauWithSn2}, we further derive an outer bound for the Gaussian channel. We first derive a bound on $R_2$ based on \eqref{eq:trdm2r2}. We set $\alpha = \rho_{21} \sqrt{\frac{P_2}{P_1}}$ and $\beta= \rho_{2s} \sqrt{\frac{P_2}{Q}}$, where $\rho_{21} = \frac{\frac{1}{n} \sum_{i=1}^n E(X_{1i}X_{2i})}{\sqrt{P_1P_2}}$ and $\rho_{2s} = \frac{\frac{1}{n} \sum_{i=1}^n E(X_{2i}S_i)}{\sqrt{P_2Q}}$. We then obtain:
\begin{flalign}
R_2 \leqslant &\frac{1}{n} \sum_{i=1}^n h(Z_i| X_{1i}S_i) - h(Z_i| X_{1i}X_{2i} S_i)\nn\\
= &\frac{1}{n} \sum_{i=1}^n h(bX_{1i} + X_{2i} +cS_i +N_{1i}|S_i X_{1i})- \frac{1}{2} \log 2\pi e\nn\\
= &\frac{1}{n} \sum_{i=1}^n h(X_{2i} +N_{1i}- \alpha X_{1i}- \beta S_i|S_i X_{1i})- \frac{1}{2} \log 2\pi e\nn\\
\leqslant& \frac{1}{n} \sum_{i=1}^n h(X_{2i} +N_{1i}- \alpha X_{1i}- \beta S_i)- \frac{1}{2} \log 2\pi e\nn\\
= & \frac{1}{2n} \sum_{i=1}^n \log\left(2 \pi e E(X_{2i} +N_{1i}- \alpha X_{1i}- \beta S_i)^2\right)- \frac{1}{2} \log 2\pi e\nn\\
\leqslant & \frac{1}{2} \log \left( P_2 +1 +\alpha^2 P_1 + \beta^{2} Q - 2 \alpha \frac{1}{n} \sum_{i=1}^n E[X_{1i}X_{2i}] - 2\beta \frac{1}{n} \sum_{i=1}^n E[X_{2i}S_i]\right)\nn\\
=& \frac{1}{2} \log ( 1 +(1 -  \rho_{2s}^2 -  \rho_{21}^2) P_2 )\label{eq:GauWithSn2ConR2}
\end{flalign}


Following \eqref{eq:gauss1r12}, we obtain the following bound on $R_1 + R_2$ based on \eqref{eq:trdm2r121}
\begin{flalign}
&R_1 + R_2 \leqslant \frac{1}{2} \log \left(b^2 P_1 + P_2 +1 + 2b \rho_{21}\sqrt{P_1P_2}-  \rho_{2s}^2 P_2 \right).
\end{flalign}

We further derive \eqref{eq:trdm2r122} for the Gaussian channel as follows:
\begin{flalign}
R_1 +R_2 & \leqslant \frac{1}{n} \sum_{i=1}^n [I(X_{1i};Y_i) + I(X_{2i};Y_i|X_{1i}S_i) ] \nn\\
& = \frac{1}{n} \sum_{i=1}^n [h(Y_i)- h(Y_i|X_{1i}) + h(Y_i|X_{1i}S_i)  - h(Y_i|S_iX_{1i}X_{2i})] \nn\\
& = \frac{1}{n} \sum_{i=1}^n [h(Y_i)- I(Y_i;S_i|X_{1i}) - h(Y_i|S_iX_{1i}X_{2i})] \nn \\
& = \frac{1}{n} \sum_{i=1}^n [h(Y_i)- h(S_i)+ h(S_i|X_{1i} Y_i) - h(Y_i|S_iX_{1i}X_{2i})] \label{eq:gauss2r122}
    \end{flalign}

Following \eqref{eq:gauss1r1-1}, \eqref{eq:gauss1r1-2}, \eqref{eq:gauss1r1-3}, and \eqref{eq:gauss1r1-4} in Appendix \ref{apx:ProofGauWithSn1Con}, we obtain
\begin{equation}\nonumber
  \begin{split}
    &\frac{1}{n} \sum_{i=1}^n h(Y_i)\leqslant \frac{1}{2} \log 2\pi e(P_1 + a^2P_2 + Q +1 + 2a\rho_{21}\sqrt{P_1P_2} + 2a\rho_{2s}\sqrt{P_2Q})\\
    &\frac{1}{n} \sum_{i=1}^n h(Y_i|X_{1i} X_{2i} S_i) =  \frac{1}{2} \log 2\pi e\\
    &\frac{1}{n} \sum_{i=1}^n h(S_i) = \frac{1}{2} \log 2\pi e Q\\
    &\frac{1}{n} \sum_{i=1}^n h(S_i|X_{1i} Y_i) \leqslant  \frac{1}{2} \log 2\pi e \frac{(a^2(1-\rho_{21}^2 - \rho_{2s}^2)P_2 +1)Q}{a^2(1-\rho_{21}^2)P_2 +2a\rho_{2s}\sqrt{P_2Q}+Q+1}
  \end{split}
\end{equation}

Substituting the above bounds into \eqref{eq:gauss2r122}, we obtain
\begin{flalign}
R_1 +R_2 \leqslant & \frac{1}{2} \log\left(1+\frac{P_1+2a\rho_{21}\sqrt{P_1P_2}+a^2\rho_{21}^2P_2}{a^2(1-\rho_{21}^2)P_2+2a\rho_{2s}\sqrt{P_2Q}+Q+1}\right) \nn\\
&+\frac{1}{2}\log\left( 1+ a^2(1 -\rho_{2s}^2 -  \rho_{21}^2) P_2 \right)
\end{flalign}
which concludes the proof.


\renewcommand{\baselinestretch}{1}
\begin{small}

\bibliographystyle{unsrt}
\bibliography{ICstate_journal}

\end{small}



%

%
%
%


\end{document}